\UseRawInputEncoding
\documentclass[11pt,a4paper]{article}

\usepackage{jheppub}
\usepackage{enumitem}
\usepackage{amsmath,amsthm}
\usepackage{braket}
\usepackage{amsfonts}
\usepackage{lscape}
\usepackage{dsfont}
\usepackage{tikz}
\usepackage{tikz-cd}
\usepackage{romannum}
\usepackage{blkarray}

\newlist{steps}{enumerate}{1}
\setlist[steps, 1]{label = \textbf{Step \arabic*:}}

\newtheorem*{prop*}{Proposition}
\newtheorem{prop}{Proposition}
\newtheorem*{thm*}{Theorem}
\newtheorem{thm}{Theorem}
\newtheorem{lem}{Lemma}
\theoremstyle{definition}

\def\be{ \begin{equation} }
\def\ee{ \end{equation} }
\def\bs{ \begin{split} }
\def\es{ \end{split}}
\def\bea{\begin{eqnarray}}
\def\eea{\end{eqnarray}}
\def\beas{\begin{align*}}
\def\eeas{\end{align*}}
\def\be{\begin{equation}}
\def\ee{\end{equation}}

\def\Aut{{\rm Aut}}

\def\Co0{{\rm Co}_0}

\def\det{{\rm det}}

\def\Fix{{\rm Fix}}

\def\la{\langle}
\def\ra{\rangle}

\newcommand{\cln}{{\,:\,}}

\def\one{{\hbox{ 1\kern-.8mm l}}}

\def\bs{\bar{s}}

\def\CD {{\cal D}}
\def\CE {{\cal E}}

\def\CH {{\cal H}}

\def\CN {{\cal N}}

\def\CE {{\cal E}}

\def\CH {{\cal H}}

\def\CQ {{\cal Q}}

\def\CT {{\cal T}}

\def\IZ{{\mathbb{Z}}}

\title{Conway Subgroup Symmetric Compactifications Redux}

\author[]{Zihni Kaan Baykara,}
\author[]{Jeffrey A. Harvey}

\affiliation[]{Enrico Fermi Institute and Department of Physics \\
$~~ $University of Chicago \\
$~~ $933 East 56th Street, Chicago IL 60637, U.S.A.}

\emailAdd{zkbaykara@uchicago.edu}
\emailAdd{j-harvey@uchicago.edu} 

\abstract{We extend the investigation in \cite{Harvey:2017xdt} of special toroidal compactifications of heterotic string theory for which the half-BPS states provide representations of subgroups of the Conway group. We also explore dual descriptions of these theories and find that they are all linked to either F-theory or type IIA string theory on K3 surfaces with symplectic automorphism groups that are the same Conway subgroups as those of the heterotic dual. The matching with type IIA K3 dual theories includes both the matching of symmetry groups and a comparison between the Narain lattice on the heterotic side and the cohomology lattice on the type IIA side. We present twelve examples where we can identify a type IIA dual K3 orbifold theory as the dual description of the heterotic theory. In addition, we include a Mathematica package that performs most of the computations required for these comparisons. 
}


\arxivnumber{}

\begin{document}
\maketitle
\flushbottom

\section{Introduction}\label{sec:Intro}

This paper is concerned with the construction of special toroidal compactifications of heterotic string theory and their dual descriptions. These compactifications are special because they have interesting subgroups of the Conway group that preserve supersymmetry and act on the spectrum of BPS states. This work extends that of \cite{Harvey:2017xdt}, which was motivated by Mathieu moonshine and mainly focused on orbifold constructions leading to dual pairs involving heterotic string on $K3 \times T^2$ and IIA string on Calabi-Yau manifolds. See also \cite{Banerjee:2020szx} where the lattice techniques used in \cite{Harvey:2017xdt} are used to relate hyperk\"ahler isometry groups of K3 manifolds to subgroups of the Conway group.

Our focus here is directly on special toroidal heterotic compactifications and a detailed study of their dual description in terms of type IIA string on K3 sigma models. There are several motivations for this work. First, while duality between heterotic string on $T^4$ and type IIA string on K3 is now well understood on a general level, there are not many examples where a detailed matching has been carried out, particularly for examples with large supersymmetry preserving automorphism groups.
Second, there are conjectures in the literature that exact symmetries of string theory are gauge rather than global symmetries. See the introduction to \cite{Banks:2010zn} for a discussion as well as references to earlier literature and \cite{Harlow:2018tng} for a discussion in the context of AdS/CFT. However, the arguments for these conjectures in the case of finite groups are not very strong and to our knowledge few attempts have been made to demonstrate this in explicit examples, for example by demonstrating that there is a point in moduli space where the finite group is embedded into a continuous gauge group or by constructing the co-dimension two defects required by this claim. The examples constructed here should provide a useful testing ground for these conjectures. 

Following this introduction, section \ref{sec:Construct} provides the necessary lattice theoretic background for construction of special lattices and describes several explicit Narain lattices with Conway subgroup symmetry. We also provide a description of how to use the associated Mathematica package to construct these and other similar lattices. Section \ref{sec:Duals} is concerned with dual descriptions of the these heterotic theories. We make a few brief remarks about F-theory duals and then turn to type IIA duals involving special K3 surfaces. We provide evidence that specific orbifold limits of K3 surfaces arise in the dual descriptions and provide evidence for this based on the study of the cohomology lattice of the dual K3. We also describe a Mathematica package which is useful for computing the lattice $H^{\mathrm{even}}(X, \mathbb{Z})$ for our K3 surfaces $X$. Finally, appendix \ref{app:A} contains additional details on the lattice constructions we use and describes the algorithms used in the Mathematica package associated with this paper, followed up by appendix \ref{app:B} with a demonstration of a lattice construction to show how the data provided by the package is utilized.

\section*{Acknowledgements}
We would like to thank M.~Gaberdiel and R.~Volpato for helpful discussions and correspondence. JH thanks G. Moore and W. Taylor for discussions on the construction of the Narain lattice resulting from the HM44 sublattice of the Leech lattice and we thank the Aspen Center for Physics, which is supported by National Science Foundation grant PHY-1607611, for facilitating these discussions. We acknowledge support from the National Science Foundation\footnote{Any opinions, findings, and conclusions or recommendations expressed in this material are those of the author(s) and do not necessarily reflect the views of the National Science Foundation.} under grant PHY 1520748.

\section{Construction of Narain lattices with Conway subgroup symmetry}\label{sec:Construct}

Consider the heterotic string compactified on the torus $T^{8-d}$. It is known that the heterotic model on $T^{8-d}$ is characterized by an even unimodular lattice $\Gamma$ with signature $(24-d,8-d)$ \cite{Narain:1985jj, Narain:1986am}. The torus models we investigate are those that arise from a lattice $\Gamma$ constructed by gluing together the orthogonal complement of a sublattice $\mathfrak{F}_L$ of the Leech lattice $\Lambda$ with that of an isometric $E_8$ sublattice.

\subsection{Lattice theory review}\label{sec:Lattice Theory}
\subsubsection{Definitions}
First, we present a short review of the relevant definitions from lattice theory. A \emph{lattice} $(\mathfrak{L},\la,\ra)$ is a freely generated abelian group $\mathfrak{L}$ equipped with a bilinear form
\be
\la ,\ra : \mathfrak{L} \times \mathfrak{L} \rightarrow \mathbb{R}
\ee
or equivalently with a \emph{quadratic form}
$$q: \mathfrak{L} \rightarrow \mathbb{R}$$
such that
\be
q(nx)=n^2q(x)\, ,
\ee
\be \label{eq:2.3}
\frac{1}{2}[q(x+y)-q(x)-q(y)] \text{ is a bilinear form.}
\ee
The equivalence between the two definitions is given as follows.\footnote{The definitions of $\la,\ra$ and $q$ could be generalized to arbitrary fields $F$, with the exception of (\ref{eq:2.3}), in which one needs to delete the $\frac{1}{2}$ factor that we kept to make the equivalence of the definitions explicit.}
\begin{equation}\label{eq:2.4}
\begin{split}
\la , \ra \longrightarrow q(x) & := \la x,x \ra \, ,\\
q \longrightarrow \la x,y\ra & := \frac{1}{2}[q(x+y)-q(x)-q(y)] \, .
\end{split}
\end{equation}
We will sometimes suppress the bilinear form or the quadratic form of the lattice $\mathfrak{L}$. We will use $\mathfrak{L}(n)$ to denote the lattice with its quadratic form multiplied by $n$. We will refer to the minimal number of generators that generate the group as the \emph{rank} of the lattice. 

A lattice $\mathfrak{L}$ is \emph{integral} if the bilinear form $\la,\ra$ takes values in $\mathbb{Z}$, and it is \emph{even} if the quadratic form takes values in $2\mathbb{Z}$. Notice that evenness implies integrality by \eqref{eq:2.4}. The \emph{signature} $(n_0,n_-,n_+)$ of the quadratic form $q$ denotes the indices of inertia of the bilinear form considered as a symmetric square matrix on some choice of generators. We will be only considering \emph{nondegenerate} forms, i.e. $n_0=0$. Quadratic forms with $n_-=0$ (resp. $n_+=0$) are \emph{positive definite} (resp. \emph{negative definite}), and those with both $n_- \neq 0$ and $n_+ \neq 0$ are \emph{indefinite}.

An \emph{isometry} $\psi$ between two lattices $(\mathfrak{L},q)$ and $(\mathfrak{L}',q')$ is a group isomorphism such that the following diagram commutes.
\be
\begin{tikzcd}[row sep=tiny]
\mathfrak{L} \arrow[dd, "\psi"] \arrow[rd, "q"] &    \\
                         &  \mathbb{R} \\
\mathfrak{L}' \arrow[ru, "q'"]            &   
\end{tikzcd}
\ee
Isometric lattices are denoted as $\mathfrak{L}\cong \mathfrak{L}'$.

The \emph{dual lattice} $\mathfrak{L}^\vee$ of a lattice $\mathfrak{L}$ is the set of $v$ in the $\mathbb{Q}$-span of $\mathfrak{L}$ such that $\la v,x \ra \in \mathbb{Z}$ for every element $x \in \mathfrak{L}$.\footnote{$\mathbb{Q}$-span of $\mathfrak{L}$ is taken by considering $\mathfrak{L}$ as a $\mathbb{Z}$-module.} The lattice $\mathfrak{L}^\vee$ is endowed with the $\mathbb{Q}$-linear extension of the bilinear form $\la,\ra$ on $\mathfrak{L}$. If the lattice is dual to itself, it is called \emph{unimodular} or \emph{self-dual}.

Notice that when $\mathfrak{L}$ is integral, we have $\mathfrak{L} \subset \mathfrak{L}^\vee$. In fact, $v+\mathfrak{L}\subset \mathfrak{L}^\vee$ for any $v\in \mathfrak{L}^\vee$. In this case, a natural question to ask is how many $v$ there are in $\mathfrak{L}^\vee$ that generate distinct sublattices $v+\mathfrak{L}\subset \mathfrak{L}^\vee$. In other words, how should one glue copies of $\mathfrak{L}$ to get $\mathfrak{L}^\vee$? The answer to this question is given by the \emph{discriminant group}, defined as 
\be
\CD (\mathfrak{L}) :=  \mathfrak{L}^\vee / \mathfrak{L} \, .
\ee
We fix a \emph{glue vector} $r\in\mathfrak{L}^\vee$ from each equivalence class $[r]\in\CD(\mathfrak{L})$. With a choice of glue vectors, we can express $\mathfrak{L}^\vee$ in terms of $\mathfrak{L}$ as
\be\label{eq:2.7}
\mathfrak{L}^\vee = \coprod_{[r] \in \CD (\mathfrak{L})} r+\mathfrak{L}\, .
\ee

Let $\mathfrak{L}$ be an even lattice. 
Then we can endow $\CD (\mathfrak{L})$ with a quadratic form $\bar q$ which descends from that on $\mathfrak{L}^\vee$ by
\begin{align*}
\bar q: \CD (\mathfrak{L}) \rightarrow \mathbb{Q}/2\mathbb{Z}\, ,
\end{align*}
\be
\bar q([v]):= \la v,v\ra \pmod{2} \, ,
\ee
where $[v]$ denotes the equivalence class of $v\in\mathfrak{L}^\vee$ in $\CD(\mathfrak{L})$. Notice that the quadratic form $\bar q$ is well defined, since for $x\in \mathfrak{L}$ and $v\in \mathfrak{L}^\vee$, we have

$$q(v+x)=2\la v,x\ra + q(v) + q(x) \equiv q(v) \pmod{2}$$
by definition of $\mathfrak{L}^\vee$ and evenness of $\mathfrak{L}$.

The \emph{orthogonal complement} of a sublattice $\mathfrak{G} \subset \mathfrak{L}$ is defined as 
\be
\mathfrak{G}^\perp:=\{\, y\in\mathfrak{L} \mid \la y,x \ra =0 \text{ for all }x\in\mathfrak{G} \,\}\, .
\ee
A sublattice $\mathfrak{G}\subset\mathfrak{L}$ is \emph{primitive} if $\mathfrak{L}/\mathfrak{G}$ as an abelian group is free. In other words, if $\mathfrak{L}$ is a rank $d$ lattice, then a primitive rank $k$ sublattice $\mathfrak{G}$ is such that $\mathfrak{L}$ can be generated by $\mathfrak{G}$ and $d-k$ many elements in $\mathfrak{L}-\mathfrak{G}$.

\subsubsection{Useful facts}
The following is a useful characterization of primitive sublattices in even unimodular lattices.
\begin{prop}[\cite{Nikulin}]\label{prop:primitive}
The sublattice $\mathfrak{G}$ of an even unimodular lattice $(\mathfrak{L},q)$ is primitive if and only if there is an isometry
\be\label{eq:2.10}
\bar \psi: (\CD(\mathfrak{G}),\bar q) \rightarrow (\CD(\mathfrak{G}^\perp) , -\bar q)\,.
\ee
\end{prop}

Conversely, one can start with two even lattices and construct an even unimodular lattice by what is called the \emph{gluing construction}. 
\begin{lem}[Gluing Lemma]\label{lem:gluing lemma}
If $(\mathfrak{L}_1,q_1)$ and $(\mathfrak{L}_2,q_2)$ are even lattices with an isometry
\be\label{eq:2.11}
\bar \psi: (\CD(\mathfrak{L}_1),\bar q_1) \rightarrow (\CD(\mathfrak{L}_2) , \bar q_2) \, ,
\ee
then the lattice
\be
\Gamma := \{\,(x,y) \mid \bar \psi([x])=[y]\,\} \subset \mathfrak{L}_1^\vee \oplus \mathfrak{L}_2^\vee
\ee
equipped with the quadratic form
\be
q(x,y):= -q_1(x) + q_2(y)
\ee
is even and unimodular.
\end{lem}
\begin{proof}
Consider the lattices $(\mathfrak{L}_1^\vee,q_1)$ and $(\mathfrak{L}_2^\vee,q_2)$ together with an isometry
\be
\bar \psi: (\CD(\mathfrak{L}_1),\bar q_1) \rightarrow (\CD(\mathfrak{L}_2) , \bar q_2)\,.
\ee
Construct the lattice $\Gamma$ by gluing $\mathfrak{L}_1^\vee$ and $\mathfrak{L}_2^\vee$ along their isometric glue vectors as follows:\footnote{By $\bar \psi(r)$, we denote a choice of a glue vector in each $\bar \psi([r]) \in \mathcal{D}(\mathfrak{L}_2)$.}
\be\label{eq:2.12}
\Gamma := \coprod_{[r]\in \CD (\mathfrak{L}_1)}(r,\bar \psi(r))+\mathfrak{L}_1 \oplus \mathfrak{L}_2 \subset \mathfrak{L}_1^\vee \oplus \mathfrak{L}_2^\vee\,,
\ee
and equip it with the quadratic form
\be
q(x,y):=-q_1(x)+q_2(y)\,.
\ee
We can show that $q$ is even by construction. Choose an arbitrary $(x,y)\in \Gamma$. By \eqref{eq:2.7}, there is a glue vector $r$ such that
\be
(x,y)=(r, \bar \psi(r))+(v_1,v_2)
\ee
with some $v_1 \in\mathfrak{L}_1$, $v_2\in\mathfrak{L}_2$, and $[x]=[r]$. Then,
\be
q(x,y)\equiv -\bar q_1([r]) + \bar q_2( \bar \psi([r])) =0 \pmod{2}\,.
\ee
To show that $\Gamma$ is unimodular, choose $(u,v)\in \Gamma^\vee \subset (\mathbb{Q} \otimes \mathfrak{L}) \oplus (\mathbb{Q} \otimes \mathfrak{L})$. For all $(x,y)\in \Gamma$, we have
\be\label{eq:2.18}
\la (u,v),(x,y) \ra = -\la u,x \ra + \la v,y \ra \in \mathbb{Z}\,,
\ee
where the bilinear product is induced by $q$. Taking $x$ (resp. $y$) to be zero in \eqref{eq:2.18} implies $u\in \mathfrak{L}_1^\vee$ (resp. $y\in\mathfrak{L}_2^\vee$). Therefore, we get $\Gamma^\vee \subset \mathfrak{L}_1^\vee \oplus \mathfrak{L}_2^\vee$. Now we can use the bilinear product taking values in $\mathbb{Q}/\mathbb{Z}$, induced by $\bar q$ on $\CD(\Gamma)$ as
\be
-\la [u],[x] \ra + \la [v],[y] \ra \equiv 0 \pmod{1}\,.
\ee
Since $\bar \psi$ and the quadratic forms commute, we get
\be
-\la \bar \psi([u]),[y] \ra + \la [v],[y] \ra = \la - \bar \psi([u]) + [v] , [y] \ra \equiv 0 \pmod{1}\,.
\ee
We see that $\bar \psi([u])=[v]$ since the bilinear form is nondegenerate. We conclude that $\Gamma^\vee=\Gamma$.
\end{proof}

The classification of even unimodular lattices helps us determine the lattice obtained by the gluing construction. In particular, indefinite even unimodular lattices are unique up to isometry.
\begin{prop}[\cite{Gerstein}]\label{prop:2}
If $\mathfrak{L}$ is an indefinite even unimodular lattice with signature $(n_-,n_+)$, then $n_+ - n_- \equiv 0 \pmod{8}$, and
\be
\mathfrak{L} \cong E_8(\pm1)^{\oplus \frac{|n_+-n_-|}{8}} \oplus U^{\oplus \min(n_+,n_-)} =: \mathrm{II}^{n_-,n_+}\,,
\ee
where $(\pm 1)=\mathrm{sign}(n_+ -n_-)$, and $U$ is the hyperbolic lattice with the bilinear form $\begin{pmatrix}
0 & 1 \\
1 & 0
\end{pmatrix}$ defined on some choice of generators that we call the standard basis of $U$.
\end{prop}
The classification of positive definite even unimodular lattices is not as precise except for low rank where we have the following result. These lattices can only have rank that is a multiple of 8. We have the $E_8$ lattice for rank 8, $E_8^{2}$ and $D_{16}^+$ for rank 16, and the 24 Niemeier lattices in rank 24 in which the Leech lattice $\Lambda$ is the only one with no roots.

Lastly, we describe a procedure for obtaining a natural set of generators for integral lattices making their discriminant group structure manifest. Since for integral lattices we have $\mathfrak{L}\subset \mathfrak{L}^\vee$, one can express the generators $x_i$ of $\mathfrak{L}$ as a $\mathbb{Z}$-linear combination of generators $v_i$ of $\mathfrak{L}^\vee$. Let $\mathfrak{L}$ be a rank $n$ lattice, then we can define

\be
X := \begin{pmatrix}
x_1 \\
\vdots \\
x_n
\end{pmatrix}\,, \qquad V := \begin{pmatrix}
v_1 \\
\vdots \\
v_n
\end{pmatrix}\,,
\ee
so that
\be\label{eq:2.24}
X=GV
\ee
for some matrix $G$ with integer entries. In fact, $G$ is the \emph{Gram matrix} $G:=XX^T$ of the set of basis $x_i$, where multiplication of the matrix elements is given by the bilinear form \cite{Gerstein}. Furthermore, we will use the following algebraic fact.\footnote{The proposition is in fact true for all matrices with entries in a PID.}

\begin{prop}[Smith Decomposition \cite{Newman, DQ}]\label{prop:Smith}
If $M$ is a matrix with integer entries, then there are unimodular ($\mathbb{Z}$-invertible) matrices $P,Q$ such that $D=PMQ$ is a diagonal matrix with integer entries. Furthermore, denoting the $i$th diagonal element of $D$ by $d_i$, we have $d_{i} \mid d_{i+1}$ for all $i$.
\end{prop}

Using Prop. \ref{prop:Smith} as applied to eq. \eqref{eq:2.24}, we get
\begin{align}
PX & = (PGQ)(Q^{-1}V)\\
\label{eq:2.14}
\tilde X & = D \tilde V\,,
\end{align}
where $\tilde X=PV$ and $\tilde V=Q^{-1}V$ are new sets of basis for $\mathfrak{L}$ and $\mathfrak{L}^\vee$ respectively, since the unimodular matrices $P,Q^{-1}$ act as a basis change on $X$ and $V$. Another way of stating \eqref{eq:2.14} is that for all integral lattices $\mathfrak{L}$, one can find a set of generators $\tilde x_i$ of $\mathfrak{L}$ and a set of generators $\tilde v_i$ of $\mathfrak{L}^\vee$ such that $\tilde x_i$ is an integer multiple of $\tilde v_i$:
\be\label{eq:2.15}
\tilde x_i = d_i \tilde v_i\,, \quad d_i\in\mathbb{Z}\,.
\ee
The form \eqref{eq:2.14}, or equivalently the set of basis \eqref{eq:2.15}, is called the \emph{Smith Normal Form} of the lattice. The relationship between the generators of the lattice and its dual immediately implies that the discriminant group $\CD(\mathfrak{L})$ is
\be \label{eq:2.16}
\CD(\mathfrak{L}) \cong \mathbb{Z}_{d_1} \times \cdots \times \mathbb{Z}_{d_n}\,.
\ee
Indeed, representing a lattice $\mathfrak{L}$ with its Smith Normal Form is the most natural way to make the discriminant group structure explicit.
%------------------------------------------------------------------------------------------------

\subsection{Construction of Narain lattices}\label{sec:Construct Theory}
We now describe the construction of the even unimodular lattice $\Gamma$ with signature $(24-d,8-d)$, starting with a primitive rank $d$ sublattice $\mathfrak{F}_L\subset \Lambda$ and an isometric primitive copy $\mathfrak{F}_R\subset E_8$. We follow \cite{Harvey:2017xdt} in our review. 

Suppose we are given the isometry
\be
\psi_{LR}:(\mathfrak{F}_L,q_L) \rightarrow (\mathfrak{F}_R,q_R)
\ee
between primitive sublattices $\mathfrak{F}_L\subset \Lambda$ and $\mathfrak{F}_R \subset E_8$, with their quadratic forms given by the embeddings $\Lambda\hookrightarrow \mathbb{R}^{24}$ and $E_8\hookrightarrow \mathbb{R}^8$ respectively.\footnote{We could generalize the construction by only requiring an isometry of discriminant groups.} We can linearly extend $\psi_{LR}$ to an isometry of the duals, which factors through and gives an isometry of the discriminant groups, $\bar\psi_{LR}:(\CD(\mathfrak{F}_L),\bar q_L) \rightarrow (\CD(\mathfrak{F}_R),\bar q_R)$. 

Using \eqref{eq:2.10}, we also get isometries between the discriminant groups of the orthogonal lattices
\be
\begin{split}
\bar\psi_L: (\CD(\mathfrak{F}_L^\perp),\bar q_L) \rightarrow (\CD(\mathfrak{F}_L),-\bar q_L)\,,\\
\bar\psi_R: (\CD(\mathfrak{F}_R),-\bar q_R) \rightarrow (\CD(\mathfrak{F}_R^\perp),\bar q_R)\,.
\end{split}
\ee 
Notice that $\bar\psi_{LR}$ equivalently induces an isometry with the signs of both quadratic forms inverted. We can compose all of these isometries to get an isometry of discriminant groups of orthogonal complements of $\mathfrak{F}_R,\mathfrak{F}_L$ as

\be
\bar\psi:= \bar\psi_R \circ \bar\psi_{LR} \circ \bar\psi_L: (\CD(\mathfrak{F}_L^\perp),\bar q_L) \rightarrow (\CD(\mathfrak{F}_R^\perp),\bar q_R)\,.
\ee 

Now that we have two even lattices $\mathfrak{F}_L^\perp$ and $\mathfrak{F}_R^\perp$ with isometric discriminant groups, we can glue them as described in Lemma \ref{lem:gluing lemma} to get the even unimodular lattice
\be\label{eq:2.31}
\Gamma := \coprod_{[r]\in \CD(\mathfrak{F}_L^\perp)} (r,\psi(r))+ (\mathfrak{F}_L^\perp,\mathfrak{F}_R^\perp)\,. \ee
Notice that the lattice indeed has signature $(24-d,8-d)$ since it is contained in $\left(\mathfrak{F}_L^\perp\right)^\vee \oplus \left(\mathfrak{F}_R^\perp\right)^\vee$, which is embedded in $\mathbb{R}^{24-d,0}\oplus\mathbb{R}^{0,8-d}$. Therefore, $\Gamma$ is isometric to $\mathrm{II}^{24-d,8-d}$.

The embedding of $\Gamma$ in $\mathbb{R}^{24-d,8-d}$ is special in the way that it contains the $O(n)$ automorphisms of $\Lambda$ and $E_8$ that fix $\mathfrak{F}_L$ and $\mathfrak{F}_R$, respectively. We denote these subgroups as $\Fix(\mathfrak{F}_L)\subset \Co0$ and $\Fix(\mathfrak{F}_R)\subset W(E_8)$, where $\Co0$ is the automorphism group of the Leech lattice and $W(E_8)$ is the Weyl group of the $E_8$ lattice, generated by the reflections in the hyperplanes orthogonal to its roots. 

Take $x \in \left(\mathfrak{F}_L^\perp\right)^\vee$ and $g_L \in \Fix(\mathfrak{F}_L)$. By Prop. \ref{prop:primitive}, there is an $x'\in \mathfrak{F}_L^\vee$ such that $x+x' \in \Lambda$. Since $g_L$ fixes $x'$,
\be g_L\cdot x - x = g\cdot (x+x') - (x+x') \in \Lambda\,.\ee
Also by considering $g_L$ as a map in $\mathfrak{F}_L^\perp$ and extending it to the dual lattice we have
\be 
g_L\cdot x - x \in \left(\mathfrak{F}_L^\perp\right)^\vee\,.
\ee
Therefore $g_L\cdot x-x$ belongs to the intersection $\left(\mathfrak{F}_L^\perp\right)^\vee \cap\Lambda$,
\be 
g_L\cdot x - x \in \mathfrak{F}_L^\perp\,.
\ee 
This means that $g_L$ preserves which copy of $\mathfrak{F}_L^\perp$ our $x$ belongs to in $\left(\mathfrak{F}_L^\perp\right)^\vee$, i.e. $g_L([x])=[x]$. Similar arguments apply for the right side with $y$ in $\left(\mathfrak{F}_R^\perp\right)^\vee$ and $g_R\in\Fix(\mathfrak{F}_R)$. We conclude that $g_L \times g_R$ induces an action on $\Gamma$ and

\be \label{eq:2.35}\Fix(\mathfrak{F}_L)\times \Fix(\mathfrak{F}_R) \subset \Aut(\Gamma)\,.\ee

Inspired by \eqref{eq:2.35}, the points in the heterotic Narain moduli space that correspond to $\Gamma$ are called \emph{Conway Subgroup Symmetric Compactifications}, or CSS compactifications for short \cite{Harvey:2017xdt}. We refer to the lattices $\Gamma$ as \emph{CSS lattices}, and sometimes by the Leech sublattice HM\# from which they were constructed (see section \ref{sec:Computation}).

%------------------------------------------------------------------------------------

\subsection{Computation using fixed-point sublattices}\label{sec:Computation}

In this section, we provide an overview of our supplemental Mathematica package for computing concrete examples of CSS lattices.

Our starting point is to find a sublattice $\mathfrak{F}_L$ of the Leech lattice $\Lambda$ with some nontrivial $\Fix(\mathfrak{F}_L)\subset \Co0$. Thanks to H{\"o}hn and Mason \cite{HM}, all sublattices with nontrivial fix groups have been classified up to conjugacy. They show that there are 290 such classes, and provide an invariant lattice $\mathfrak{F}_L$ and a coinvariant lattice $\mathfrak{F}_L^\perp$ from each class in their supplemental MAGMA program. We refer to their invariant lattices as HM\# where \# stands for the number of the lattice in Table 1 of their paper.

The only obstruction to the construction of CSS lattices is when $\mathfrak{F}_L$ has no primitive embedding in $E_8$. We have found a primitive embedding for all the rank $d<5$ Leech sublattices of H{\"o}hn and Mason, so there is a CSS construction on $T^n$ with $n>3$ for all cases. For rank $d=5$ sublattices (corresponding to $T^3$), we have found a primitive embedding only for HM69, HM70, HM72, HM73, HM74, HM76, HM77, HM82, HM84, HM85, HM87. For rank $d=6$ (corresponding to $T^2$), we have found an embedding only for HM44. For rank $d=7$ (corresponding to $S^1$), we have found no primitive embeddings for any of the sublattices.

In our Mathematica package, the command
\begin{verbatim}
    CSSLattice[dim,#]
\end{verbatim}
returns the generators of the desired CSS lattice as vectors in $\mathbb{R}^{20,4}$, where the first input \verb|dim| is the rank of the sublattice, and \verb|#| is the number of the lattice starting counting from the first lattice of the same rank in Table 1 of the H{\"o}hn-Mason paper. For example, for HM101 we have $\verb|dim|=4,\verb|#|=3$.

The command is not an algorithm but a database query, as some of the $E_8$ embeddings are computationally difficult to find in real time. We also provide the details of the construction such as the choice of $\mathfrak{F}_R$, and isometries between discriminant groups through various commands in the Mathematica package to make the construction of $\Gamma$ explicit. Refer to appendix \ref{app:A} for the specifics of the computations. Also, see appendix \ref{app:B} for an example CSS lattice computation that utilizes the data provided by the package.

%------------------------------------------------------------------------------------------------------------
\section{Duals of heterotic string with Conway subgroup symmetry}\label{sec:Duals}

In this section we consider dual descriptions of the Conway subgroup symmetric compactifications constructed in the previous section. There are no CSS compactifications on $S^1$ and a single one on $T^2$ constructed from the embedding of the HM44 lattice in the Leech lattice. The dual description of heterotic string on $T^2$ is F-theory on an elliptically fibered K3 surface and we discuss this in the first subsection below. CSS compactifications exist for compactifications on $T^n$ for all $n \ge 4$ for all of the sublattices in \cite{HM}. We focus on the case $n=4$ where the dual description is type IIA string theory on a K3 surface. The second subsection describes twelve examples in which we have obtained a detailed description of the dual theory and checked the duality both by comparison of the supersymmetry preserving symmetries and by comparison of the Narain lattice on the heterotic side with the cohomology lattice of the K3 surface on the type II side.

In this section, we use ATLAS notation \cite{ATLAS} to denote the group structure, where $G\cln H$ denotes the semidirect product and $G.H$ is a group with a normal subgroup $G$ such that $G.H/G\cong H$.

\subsection{Heterotic/F-theory dual}

Among rank 6 sublattices of the Leech lattice, only HM44 has a primitive embedding in $E_8$. This can also be deduced without the Mathematica computations by using \cite[Theorem 1.12.2]{Nikulin}, which states that if $\mathfrak{L}$ is a rank $d$ lattice that can be primitively embedded in $E_8$ and the discriminant group $\mathcal{D}(\mathfrak{L})$ has minimal number of generators $\ell$, then $8-d \geq \ell$. For $d=6$, only HM44 satisfies this condition with $\mathcal{D}(\mathfrak{L})\cong \mathbb{Z}_{10}\times \mathbb{Z}_{30}$.

Constructing the CSS lattice $\Gamma$ corresponding to HM44 as explained in appendix \ref{app:B}, we get a heterotic theory on $T^2$, which is dual to F-theory on an elliptically fibered K3 surface. To get the supersymmetry preserving symmetries, we consider the subgroup of automorphisms $G_R\subset \mathrm{Aut}(\Gamma)$ that fixes the right side of the Narain lattice, which turns out to be $A_5$.\footnote{Note that automorphisms of $\Gamma$ that fix the right side form a group containing $\mathrm{Fix}(\mathfrak{F}_L)$ and is possibly larger. In our case though, a calculation in MAGMA shows that $G_R= \mathrm{Fix}(\mathfrak{F}_L)\cong A_5$.} 

The $A_5$ symmetry we found on the heterotic side should correspond in the dual theory to a symplectic action on the K3 surface $X$, i.e. the automorphism $g\in\mathrm{Aut}(X)$ should act trivially on the canonical line bundle as $g^*\omega_X=\omega_X$. Using Xiao's list \cite[Table 2]{Xiao} of all possible symplectic automorphism groups of K3 surfaces, we can determine some candidate K3 surfaces for the duality.\footnote{Xiao's list is a stronger classification than Mukai's theorem \cite{Mukai}.} The symplectic K3 groups that contain $A_5$ together with their numbers in Xiao's table are as follows: $A_5$ (\#55), $S_5$ (\#70), $A_6$ (\#79), and $M_{20}=2^4\cln A_5$ (\#81). 

We can discard surfaces with $S_5$ (\#70) or $A_6$ (\#79) symplectic symmetry group since the action of $A_5$ in those is only a subgroup of all automorphisms, and they have no counterpart on the heterotic side of the duality. We also tentatively discard surfaces with the symplectic automorphism group $2^4 \cln A_5$ (\#81) where the $A_5$ acts on a group isomorphic to $2^4$, for which we could not find a counterpart on the heterotic side.\footnote{The reader can refer to \cite{Sarti} for the classification of such surfaces.}

Our arguments lead us to only consider K3 surfaces admitting an $A_5$ (\#55) symplectic symmetry group. The symplectic symmetries do not uniquely choose a surface, but their extension to a faithful action on the surface does. The following theorem classifies all such extensions.

\begin{thm}[\cite{zhang2006}]\label{thm:zhang}
Suppose $G$ acts on a K3 surface $X$ faithfully and $A_5\trianglelefteq G$ symplectically. Then $G$ is isomorphic to one of the following groups: $A_5,\,S_5=A_5.\IZ_2,\,A_5\times \IZ_2$. Each of these groups are realizable for some K3 surface.
\end{thm}

Here, the $\IZ_2$ factor describes the action on the canonical line bundle $\omega_X$: under the map $G\to G/A_5\cong \IZ_2$, if the element $g$ is mapped to $1\in \IZ_2$, then $g^*\omega_X=\omega_X$, and if $g$ is mapped to $-1\in \IZ_2$, then $g^*\omega_X=-\omega_X$. The faithful action group $G$ determines the surface uniquely \cite[Section 4]{zhang2005}. Also, by Lemma 1.1 in \cite{zhang2006}, Picard number of these surfaces satisfy $\rho(X)\geq 19$, ensuring that they are elliptic due to the fact that K3 surfaces with $\rho(X)\geq 5$ are elliptic \cite{huybrechts2016lectures}.

We now mention the candidate surfaces dual to the heterotic theory corresponding to HM44. One explicit surface is obtained by considering the intersection of an $S_5$ symmetric quadric and a cubic, which is the K3 surface
\begin{align}\label{eq:surface}
    X=\left\{\,\sum_{i=0}^4 X_i^2 = \sum_{i=0}^4 X_i^3=0 \,\right\}\subset \mathbb{P}^4\,.
\end{align}
The group $A_5$ acts on the surface symplectically by permuting the coordinates, and $S_5=A_5. \IZ_2$ acts faithfully. To see that $A_5$ acts symplectically and $S_5$ does not, one can use Lemma 2.1 in \cite{Mukai}.

Different elliptic fibrations of the same surface can give rise to different physics, so the choice of an elliptic fibration is also an important part of the F-theory data. The elliptic fibrations we should consider have to respect the absence of gauge enhancement on the heterotic side, as there are no lattice vectors $(p_L,0)$ with $p_L^2=2$. This means that the allowed singular fibers can only be of Kodaira type $\mathrm{I}_0, \mathrm{I}_1,$ or $\mathrm{II}$ \cite{Bershadsky}.

If the Picard lattice $\mathrm{Pic}(X)$ of the surface is computable, one can use the correspondance between the null vectors and elliptic fibrations. Specifically, there is a bijective correspondance between the primitive divisors $E\in \mathrm{Pic}(X)$ that lie in the nef cone with vanishing self-intersection $E\cdot E=0$ and the distinct elliptic fibrations. Moreover, the lattice $W$ of (-2)-vectors inside the null vector's orthogonal complement $E^\perp \subset \mathrm{Pic}(X)$ describe the singular fibers. Therefore, in order to find a fibration that has no gauge enhancement, one should find a null vector $E$ such that its orthogonal complement has root lattice of type $A_1^n$ for some $n$.

An explicit elliptically fibered K3 with an $A_5$ action can be given by the Weierstrass model \cite{Shioda}
\begin{align}\label{eq:Shioda}
    X:\qquad y^2 = x^3 + t^{11}-11 t^6 -t\,.
\end{align}
This surface is obtained by placing cusp fibers ($y^2=x^3$) at the vertices of an icosahedron inscribed inside $\mathbb{P}^1\cong S^2$, and having $A_5$ act symplectically by shuffling the singular fibers. Notice that it is necessarily one of the surfaces classified in Theorem \ref{thm:zhang}. In order to determine which one it is, one would look for the non-symplectic symmetries of the surface. 

There are 12 cuspidal fibers of Kodaira type $\mathrm{II}$ on the surface \eqref{eq:Shioda}, and therefore there is no gauge enhancement as expected. These properties make \eqref{eq:Shioda} a good candidate for the F-theory dual. 

However, note that our arguments do not uniquely pick out one surface, but by virtue of Theorem \ref{thm:zhang} and the fact that there are only finitely many non-isomorphic elliptical fibrations of a surface, we can at least conclude that there are only finitely many candidates for the dual K3 surface. We presented two such explicit examples (although they might be isomorphic) and showed techniques that may be helpful in finding others. In order to find the dual, further matchings need to be carried out so that a unique surface with a fibration can be chosen out of the finitely many cadidates.

\subsection{Heterotic/type IIA duals}
In this section, we consider the type IIA duals of CSS heterotic string compactifications with $d=4$. Specifically, heterotic string theory on $T^4$ is dual to type IIA string theory on a K3 surface. On the type IIA side of this duality we will only consider geometric orbifolds $T^4/\mathbb{Z}_N$ that can be blown up to a K3 surface. We report our findings in Table \ref{Table:1}. We adopt the somewhat unusual convention that $\alpha'=1$ in this section in order to more easily use results from \cite{Hiker}.

We can describe the moduli space of non-linear sigma models on the K3 surface $X$ by a choice of positive definite 4-dimensional plane $x$ in $\mathbb{R}^{20,4}$ modulo lattice automorphisms of $\Gamma^{20,4}\cong\mathrm{II}^{20,4}$, where $\Gamma^{20,4}$ is the even homology lattice $H_{\text{even}}(X,\mathbb{Z})$ equipped with the intersection form. The embedding of this lattice in $\mathbb{R}^{20,4}$ is given by Poincar{\'e} duality $\Gamma^{20,4}\cong H_{\text{even}}(X,\mathbb{Z}) \cong H^{\text{even}}(X,\mathbb{Z}) \subset H^{\text{even}}(X,\mathbb{R})\cong \mathbb{R}^{20,4}$, where the even cohomology ring is equipped with the cup product as its bilinear form. In short, the moduli space is
\be\label{eq:3.1}
O(\Gamma^{20,4})\setminus \CT^{20,4}\,,
\ee
where $\CT^{20,4}$ is the Grassmanian of 4-planes. We identify the dual of a heterotic CSS model by choosing the $\mathbb{R}$-span of $\mathfrak{F}_R^\perp$ to be the positive definite 4-plane $x\subset H^{\text{even}}(X,\mathbb{R})$. Then the automorphisms of the orthogonal complement $x^\perp \cap \Gamma^{20,4}$ correspond to the symmetries that preserve the $\CN=(4,4)$ superconformal algebra. This way, we can detect the type II duals by comparing their $\CN=(4,4)$ symmetry groups to $\Fix(\mathfrak{F}_L)$. However, only some points in the moduli space of non-linear sigma models on K3 are solvable, among these are torus sigma model orbifolds, Gepner models, and orbifolds thereof. There are some such examples in the literature that we can identify with their HM\# duals: $D_4/\mathbb{Z}_2$ with HM99 \cite{gtvw}, $A_1^4/\mathbb{Z}_2$ with HM107, $(1)^6$ Gepner with HM101, $(2)^4$ Gepner with HM116 \cite{GHV}, and $A_4/\mathbb{Z}_5$ with HM122 \cite{VG}.

We describe our procedure for obtaining such duals, with which we reproduce some of the mentioned results above. We will be mainly working with K3 surfaces obtained by blowing up the singularities of a torus orbifold $T^4/\mathbb{Z}_N$.
% It is known that blowing down such surfaces does not change the string theory, so one can use the torus orbifolds as the target space. 
To find out when a torus orbifold is a K3 model, one can calculate the elliptic genus as in \cite{Vol}. We have checked using the classification in \cite{Vol} that the choice of orbifolding symmetry in all our models indeed produce a K3 theory.

The torus models we quote from \cite{Vol} are classified according to their $\CN=(4,4)$ symmetry groups. We borrow the naming convention for such tori in loc. cit. For example, $A_1^2A_2$ is the torus model corresponding to the choice of the positive definite plane $x$ that contains the root lattice of $\mathfrak{su}(2)^{\oplus 2}\oplus \mathfrak{su}(3)$ in the even D-brane charge lattice $H^{\mathrm{even}}(T^4,\mathbb{Z})$, see \eqref{eq:3.41}. 

With our procedure we are able to specify some of the symmetry groups in the H{\"o}hn-Mason paper \cite{HM} whose group structure was not specified. We reason as follows. After $\mathbb{Z}_N$-orbifolding, one can argue that the symmetry group contains the extraspecial group $N^{1+k}$ for some $k$ and find which torus symmetries survive the orbifolding. However, this analysis by itself does not always determine the full symmetry group. Therefore, we compute the cohomology or D-brane charge lattice to confirm the duality to the corresponding CSS lattice and also determine the symmetry group structures.

\subsubsection{The \texorpdfstring{$A_1^2A_2 / \mathbb{Z}_3$}{A12 A2/Z3} model}\label{sec:Z3 model}

We first consider the $A_1^2 A_2 $ model from section 4.4.5 of \cite{Vol} and use it to state some relevant concepts and facts. Consider the following torus model with its lattice $L$ represented as a matrix with columns $l_1,\dots,l_4$ as its generators and vanishing $B$-field.

\be \label{eq:3.2} L = \frac{\sqrt{2}}{3^{1/4}}\left(
\begin{array}{cccc}
 1 & \frac{1}{2} & 0 & 0 \\
 0 & \frac{\sqrt{3}}{2} & 0 & 0 \\
 0 & 0 & 1 & \frac{1}{2} \\
 0 & 0 & 0 & \frac{\sqrt{3}}{2} \\
\end{array}
\right)\,, \qquad B=\begin{pmatrix}
0 & 0 & 0 & 0\\
0 & 0 & 0 & 0\\
0 & 0 & 0 & 0\\
0 & 0 & 0 & 0
\end{pmatrix}\,. \ee
The winding-momentum lattice is given by 
\begin{align}\label{eq:3.3}
\Gamma_{\text{w-m}}^{4,4} &=\Big\{\,\frac{1}{\sqrt{2}}(m-Bl+l, m-Bl-l) \mid l\in L, m\in L^\vee\,\Big\}\\
\begin{split}
  &= 3^{-1/4} \Big( \mathrm{span}_\mathbb{Z} \left[(1,1),(\mathbf{i},-\mathbf{i}),(e^{2\pi \mathbf{i}/3},e^{2\pi \mathbf{i}/3}),(\mathbf{i}e^{2\pi \mathbf{i}/3},-\mathbf{i}e^{2\pi \mathbf{i}/3})\right] \\
    &\qquad\quad\oplus\, \mathrm{span}_\mathbb{Z} \left[
(\mathbf{j},- \mathbf{j}), 
(\mathbf{k}, \mathbf{k}), 
(\mathbf{j}e^{2\pi \mathbf{i}/3},−\mathbf{j} e^{2\pi \mathbf{i}/3}),
(\mathbf{k} e^{2\pi \mathbf{i}/3}, \mathbf{k} e^{2\pi \mathbf{i}/3})
\right] \Big)\,,
\end{split}\end{align}
where $\mathbf{i},\mathbf{j},\mathbf{k}$ are quaternionic elements. The $\CN=(4,4)$ symmetries of this model can be written as
\be
G = (\mathrm{U}(1)^4\times \mathrm{U}(1)^4).G_0\,,\ee
where
\be
G_0 \cong \la (e^{2\pi \mathbf{i}/3},e^{2\pi \mathbf{i}/3}),(-\mathbf{j},\mathbf{j}),(-\mathbf{i},\mathbf{i})\ra \cong \mathbb{Z}_2.(\mathbb{Z}_2\times S_3)\,.
\ee

To understand how the symmetries act on the model, recall that a torus CFT is defined in terms of currents $j^a(z)$, fermions $\psi^a(z)$, as well as their right-moving analogs, and the vertex operators $V_\lambda(z,\bar z)$ for $\lambda \in \Gamma_{\text{w-m}}^{4,4}$. The $\mathrm{U}(1)^4\times\mathrm{U}(1)^4$ part of the symmetry group is generated by the zero modes $j_0,\tilde j_0$. Therefore the interesting part of the symmetry group is $G_0$, whose elements $g=(g_L,g_R)\in G_0$ act on the CFT as
\be
\begin{split}
j^1+j^2 \mathbf{i}+j^3 \mathbf{j}+j^4 \mathbf{k} &\longrightarrow g_L\cdot (j^1+j^2 \mathbf{i}+j^3 \mathbf{j}+j^4 \mathbf{k}) \,,\\
\psi^1+\psi^2 \mathbf{i}+\psi^3 \mathbf{j}+\psi^4 \mathbf{k} & \longrightarrow g_L\cdot (\psi^1+\psi^2 \mathbf{i}+\psi^3 \mathbf{j}+\psi^4 \mathbf{k})\,, \\
V_{\lambda} &\longrightarrow \xi_g(\lambda)V_{g^{-1} \cdot \lambda}\,,
\end{split}
\ee
where the action in every case is quaternionic left multiplication and $\xi_g(\lambda)=\pm 1$. Similar actions are defined on the right with $g_R$ and $\tilde j^a,\tilde \psi^a$.

Consider orbifolding the torus theory by the symmetry $g=(e^{2\pi \mathbf{i}/3},e^{2\pi \mathbf{i}/3})$ with $g^3=1$. This produces a K3 model by calculations on the elliptic genus. It has been shown in \cite{VG} that a consistent $T^4/\mathbb{Z}_3$ orbifold contains the symmetry subgroup $3^{1+4}\cln\mathbb{Z}_2$, where $\mathbb{Z}_2$ is induced by the involution $(-1,-1)$ of the underlying torus model. The extraspecial part $3^{1+4}$ is the algebra of operators defined as
\be
\lim_{z\rightarrow 0}V_\lambda (z) \ket{m,k} = e^{(k)}_\lambda \ket{m,k}, \quad k\in \mathbb{Z}_3
\,,\ee
where $k$ denotes the $g^k$-twisted sector, $m$ is some index in that sector, and the symmetries act on the vertex operators $T_{m,k}$ of the orbifolded theory by conjugation. It can be shown that there are elements $x_1,x_2,y_1,y_2 \in \Gamma^{4,4}/(1-g)\Gamma^{4,4}\cong \mathbb{Z}_3^4$ such that the symmetry group $3^{1+4}$ is generated by $e_{x_1}^{(k)},e_{x_2}^{(k)},e_{y_1}^{(k)},e_{y_2}^{(k)}$ satisfying the relations
\be
e_{x_i}^{(k)}e_{y_j}^{(k)}=\zeta^{k\delta_{ij}}e_{y_j}^{(k)}e_{x_i}^{(k)}\,,
\ee
\be
\left(e_{x_i}^{(k)}\right)^3=1=\left(e^{(k)}_{y_i}\right)^3\,,
\ee
where $\zeta=e^{2 \pi i/3}$. Each twisted sector $\CH^{(k)}$ is a 9 dimensional representation of $3^{1+4}$, and one can choose the eigenvectors of $e_{x_i}^{(k)}$ as a basis for this space so that
\begin{align}
e_{x_i}^{(k)} \ket{m_1,m_2,k} &=  \zeta^{m_i} \ket{m_1,m_2,k}\, \\
e_{y_i}^{(k)} \ket{m_1,m_2,k} &=  \ket{m_1+k\delta_{1i},m_2+k\delta_{2i},k}\,.
\end{align}

Now consider the symmetries of the $A_1^2 A_2$ model. Let $h_\mathbf{i}=(-\mathbf{i},\mathbf{i})$, $h_\mathbf{j}=(-\mathbf{j},\mathbf{j})$, and $h_\mathbf{k}=(\mathbf{k},\mathbf{k})=h_\mathbf{i}h_\mathbf{j}$. Then we have
\begin{align}\label{eq:3.13}
\begin{split}
h_\mathbf{i}^{-1}gh_\mathbf{i} & =g\,, \\
h_\mathbf{j}^{-1}gh_\mathbf{j} & =g^{-1}\,, \\
h_\mathbf{k}^{-1}gh_\mathbf{k} & =g^{-1}\,.
\end{split}
\end{align}
Therefore these symmetries survive the orbifolding. Furthermore, (\ref{eq:3.13}) implies that $h_\mathbf{i}$ preserves each twisted sector and $h_\mathbf{j},h_\mathbf{k}$ maps $g^{k}$-twisted sector to the $g^{3-k}$-twisted sector. We conclude that the symmetry group is $3^{1+4}\cln Q_8 = 3^{1+4}\cln\mathbb{Z}_2.\mathbb{Z}_2^2$, where the $\mathbb{Z}_2$ normal subgroup is generated by the involution, $\{(1,1),(-1,-1)\}\triangleleft Q_8$.

%--------------------------------------------------------------------------------------------
\subsubsection{The \texorpdfstring{$A_4 / \mathbb{Z}_2$}{A4/Z2} model}\label{sec:Z2 model}
It is known that the $D_4$ torus has the largest symmetry group as a torus model. When orbifolded by $\mathbb{Z}_2$, it gives the largest symmetry group possible \cite{gtvw} in the GHV classification of K3 sigma model symmetries \cite{GHV}, that is $2^8\cln M_{20}=\Fix(\mathrm{HM}99)$. It is natural to expect that the $\mathbb{Z}_2$ orbifold of the $A_4$ torus (the torus model with the second largest symmetry group) produces the second largest possible symmetry group $[2^9].A_5=\Fix(\mathrm{HM}100)$, where $[2^9]$ is an undetermined group of order $2^9$.

Consider the $A_4$ model from section 4.4.2 of \cite{Vol} given as
\be \label{eq:3.14} L = \frac{5^{1/4}}{\sqrt{2}}\begin{pmatrix}
1 & 1 & \frac{1}{2} & \frac{1}{2}\\
0 & 1 & \frac{1}{2} & \frac{1}{2}\\
0 & 0 & \frac{1}{\sqrt{2}} & 0\\
0 & 0 & 0 & \frac{1}{\sqrt{2}}
\end{pmatrix},\, \qquad B=\frac{1}{\sqrt{5}}\left(
\begin{array}{cccc}
 0 & -1 & -\sqrt{2} & 0 \\
 1 & 0 & 0 & -\sqrt{2} \\
 \sqrt{2} & 0 & 0 & 1 \\
 0 & \sqrt{2} & -1 & 0 \\
\end{array}
\right)\,. \ee
It has the $\CN=(4,4)$ symmetry group

\be
G = (\mathrm{U}(1)^4\times \mathrm{U}(1)^4).G_0\,,\ee
\be\label{eq:3.16}
G_0 \cong \mathbb{Z}_2.A_5\,,
\ee
where $\mathbb{Z}_2$ corresponds to the $(-1,-1)$ involution. In a similar way, all $\mathbb{Z}_2$ orbifold symmetry groups contain a $2^{1+8}$ subgroup, generated by the half-period translations of the lattice $L$ together with half-period translations of its T-dual acting on the states with charge $(p_L,p_R)$ of the torus theory by
\be
H_{(a_L,a_R)}: (-1)^{2(a_L, a_R)\cdot (p_L,p_R)}\,, \qquad (a_L,a_R)\in \left(\frac{1}{2}\Gamma_{\text{w-m}}^{4,4}\right)/\Gamma_{\text{w-m}}^{4,4} \cong \mathbb{Z}_2^8\,,
\ee
and the quantum symmetry $\CQ$ fixing the untwisted sector and acting as $(-1)$ on the twisted sector.

Since all symmetries in \eqref{eq:3.16} modulo the involution survive the orbifolding, the symmetry group of the $A_4/\mathbb{Z}_2$ model contains $2^{1+8}\cln A_5$. This analysis specifies the subgroup $[2^9]$ mentioned earlier.

%-------------------------------------------------------------------------------------------------
\subsubsection{K3 lattice computations}\label{sec:K3 lattice}
We present a review of the cohomology theory of torus orbifold K3 surfaces and how orbifolding embeds the even cohomology lattice of $T^4$ in the cohomology lattice of K3. The arguments will be cut short and mostly omitted to give just enough background to state the facts \eqref{eq:3.33} and \eqref{eq:3.34}. For the complete treatment, the reader is referred to the original paper \cite{Wen}.

The moduli space of $T^4$ and K3 theories are given by a choice of a positive definite 4-dimensional plane $x$ in the even cohomology $H^{\text{even}}(X,\mathbb{R})\cong \mathbb{R}^{4+\delta,4}$ where $\delta=0$ for tori and $\delta=16$ for K3 surfaces. The integer cohomology lattice $\Gamma^{20,4}=H^{\text{even}}(X,\mathbb{Z})$ is embedded in $H^{\text{even}}(X,\mathbb{R})$ by extending the lattice from a $\mathbb{Z}$-module to an $\mathbb{R}$-module. From Prop. \ref{prop:2}, we know that
\be
H^{\text{even}}(X,\mathbb{Z}) \cong E_8(-1)^{\oplus \delta/8} \oplus U^{\oplus 4}=\mathrm{II}^{4+\delta,4}\,.
\ee
Therefore, a theory is determined by the position of the plane $x$ in relation to $\mathrm{II}^{4+\delta,4}$, and we may as well consider the covering $\CT^{4+\delta,4}$ of the moduli space (\ref{eq:3.1}).

We describe the technology developed in \cite{Wen}. Let us begin with a torus theory $\mathbb{R}^4/L$ with lattice $L=\mathrm{span}_\mathbb{Z}(l_1,l_2,l_3,l_4)$ and background $B$. We would like to get the position of the plane $x_T$ in relation to $\Gamma^{4,4}$ and see how it determines the plane $x$ in $\Gamma^{20,4}$ after orbifolding.\footnote{We denote objects associated to the torus theory with a subscripted T.} There are two methods for obtaining $x_T$. We could have used the $\mathrm{SO}(4,4)$ triality on $\Gamma^{4,4}_{\text{w-m}}$ \cite{Vol}, but we will use a more geometric approach: there is an isomorphism
\begin{align}
\CT^{3+\delta,3}\times \mathbb{R}^+ \times H^2(X,\mathbb{R}) & \cong \CT^{4+\delta,4}\,,\\
(\Sigma,V,B)& \mapsto \mathrm{span}_{\mathbb{R}}(\xi(\Sigma),\xi_4) \label{eq:3.18}\,,
\end{align}
with the maps
\begin{align}
    \begin{split}
        \xi(\sigma)& :=\sigma-\la\sigma,B\ra v\,,\\
        \xi_4& :=v^0+B+\left(V-\frac{\la B,B\ra}{2}\right)v\,,
    \end{split}
\end{align}
where $v^0$ and $v$ satisfy $\la v,v^0 \ra =1, \la v,v\ra = \la v^0,v^0 \ra=0$, and generate $H^0(X,\mathbb{Z})$ and $H^4(X,\mathbb{Z})$ respectively. The positive definite 3-plane $\Sigma$ is the self-dual 2-forms and is interpreted as the choice of a volume 1 metric on the space $X$, $V$ is the volume of $X$, and $B$ is the Kalb-Ramond background as an element of $H^2(X,\mathbb{R})$. The triplet $(\Sigma,V,B)$ is called the \emph{geometric interpretation} of $x$. 

Now we describe the procedure to obtain $(\Sigma_T,V_T,B_T)$ given the torus lattice $L$ and background $B_T$. Consider the coordinate system in terms of the generators of the lattice $L$ as
\be\label{eq:3.20}
\lambda_1 l_1 + \lambda_2 l_2 + \lambda_3 l_3 + \lambda_4 l_4 = x_1 e_1 + x_2 e_2 + x_3 e_3 + x_4 e_4 \in \mathbb{R}^4/L\,,
\ee
where $e_i$ is the $i$th standard basis of $\mathbb{R}^4$. We can present the integral even cohomology of the torus in terms of $\lambda_i$ as
\begin{align}\label{eq:3.21}
\begin{split}
\Gamma^{4,4}=\mathrm{span}_\mathbb{Z}(d\lambda_1 \wedge d\lambda_2, d\lambda_3 \wedge d\lambda_4, d\lambda_1 \wedge d\lambda_3, d\lambda_4 \wedge d\lambda_2, d\lambda_1 \wedge d\lambda_4,d\lambda_2\wedge d\lambda_3,\\
1,d\lambda_1\wedge d\lambda_2\wedge d\lambda_3\wedge d\lambda_4)
\end{split}
\end{align}
with the bilinear form given by the wedge product
\be
\la \alpha,\beta \ra := \int \alpha \wedge \beta\,.
\ee
We can see that the generators of the lattice $\Gamma^{4,4}$ form the standard $U^{\oplus 4}$ basis in the order they are written in (\ref{eq:3.21}). 

The metric $g$ on the torus $T^4=\mathbb{R}^4/L$ is inherited from the Euclidean covering space, therefore the induced volume form used for defining the Hodge star operator is
\be
\omega = dx_1 \wedge dx_2 \wedge dx_3 \wedge dx_4\,,
\ee
and the self-dual 2-forms forming a set of basis for $\Sigma$ are
\be
dx_1 \wedge dx_2 + dx_3 \wedge dx_4,\quad dx_1\wedge dx_3 + dx_4 \wedge dx_2, \quad dx_1 \wedge dx_4 + dx_2 \wedge dx_3\,.
\ee
We choose to express all the lattice elements in terms of $d\lambda_i$ to keep the $U^{\oplus 4}$ structure explicit. We find $dx_i$ in terms of $d\lambda_i$ by taking the exterior derivative of (\ref{eq:3.20}),
\be
dx_i = \sum_{j=1}^4 (e_i\cdot l_j) d\lambda_j\,.
\ee

Next, we set $V_T=\det (L)$, as simply the volume of the torus. Finally, if $B_T$ is an antisymmetric matrix as in \eqref{eq:3.3}, one needs to convert it from $dx_i \wedge dx_j$ to $d\lambda_i \wedge d\lambda_j$ basis by
\be\label{eq:3.28}
B \longrightarrow L^T B L\,,
\ee
so that the $d\lambda_i \wedge d\lambda_j$ component is given by the matrix element $B_{ij}$.

Now that we have $(\Sigma_T,V_T,B_T)$, we describe how geometric orbifolding by $G=\mathbb{Z}_M$, $M\in\{2,3,4\}$ induces the following map:\footnote{All of the following statements are still true for $M=6$, $\hat D_n, n\in\{4,5\}$, and the binary tetrahedral group $\hat{\mathbb{T}}$, but we will not consider such orbifolds.}
\be
(\Sigma_T,V_T,B_T) \longmapsto (\Sigma,V,B) \in \CT^{20,4}\times \mathbb{R}^+ \times H^2(X,\mathbb{R})\,.
\ee
To find the image of $\Sigma_T$, we consider $H^2(T^4,\mathbb{Z})$ after orbifolding. Away from the singularities, the orbifolding $\pi:T^4 \rightarrow X$ is a degree $M$ map. Keeping the $G$-invariant elements of the cohomology $H^2(T^4,\mathbb{Z})^{G}$, one would expect that the induced map is an embedding $H^2(T^4,\mathbb{Z})^{G}\hookrightarrow H^2(X,\mathbb{Z})$. Indeed, it was shown in \cite{Ino} that
\be
\pi_*(H^2(T^4,\mathbb{Z})^G)\cong H^2(T^4,\mathbb{Z})^G(|G|)\,,
\ee
so that there is an isometric embedding
\be\label{eq:3.29}
\sqrt{|G|}H^2(T^4,\mathbb{Z})\hookrightarrow H^2(X,\mathbb{Z})\,,
\ee
with which the map $\Sigma_T\mapsto \Sigma$ is defined.

We now describe the lattice $H^2(X,\mathbb{Z})$. Along with the embedding of the torus 2-cohomology, one has the Poincar{\'e} duals of the exceptional divisors $E_s^{(j)}$ obtained from blowing up the singularities, where $s$ denotes the singular point. The lattice of exceptional divisors $\CE_{|G|}$ have the intersection form $A_1^{16},A_2^9,$ or $A_3^4\oplus A_1^6,$ for a $\mathbb{Z}_2,\mathbb{Z}_3,\mathbb{Z}_4$ orbifold respectively. We have $\CE_{|G|} \perp \sqrt{|G|}H^2(T^4,\mathbb{Z})$, since the exceptional divisors have volume zero. In the case of $G=\mathbb{Z}_2$, the torus lattice $\sqrt{|G|}H^2(T^4,\mathbb{Z})$ and the exceptional divisors are both primitive lattices, therefore one can use Lemma \ref{lem:gluing lemma} to get a lattice with signature (19,3) isometric to $H^2(X,\mathbb{Z})$. However, this is not the case for other $G=\mathbb{Z}_M$, and one needs to consult to Prop. 2.1 of \cite{Wen} for the set $M_{|G|}$ of generators not in $\CE_{|G|}$ and $\sqrt{|G|}H^2(T^4,\mathbb{Z})$, as well as the set of generators of the primitive lattice $\Pi_{|G|}$ that contains $\CE_{|G|}$.

We can extend the embedding $\pi_*$ to the whole even cohomology by similar arguments
\be
\hat\pi: \sqrt{|G|}H^{\text{even}}(T^4,\mathbb{Z})\hookrightarrow H^{\text{even}}(X,\mathbb{Z})\,,
\ee
which gives the desired mapping $x_T \rightarrow x$. Though, we still need to express $H^{\text{even}}(X,\mathbb{Z})$ in terms of exceptional divisors and torus cohomology elements for this mapping to be meaningful.

The element $\sqrt{|G|}v$ still generates $H^4(X,\mathbb{Z})$, but $\sqrt{|G|}v^0$ Poincar{\'e} dual to $T^4-\{s\mid s\text{ is singular}\}$ does not, as there are exceptional divisors that it does not take into account. Considering a $B$-field induced solely by orbifolding offsets this issue. To describe the $B$ field we define
\be
B_{|G|}:\, \frac{1}{|G|}\la B_{|G|},E_s^{(j)} \ra = \frac{n_s^{(j)}}{|G'|}\,,
\ee
where $s$ is a $G'\subset G$ type fixed point and $\sum_j n_s^{(j)}E_s^{(j)}$ is the highest root in the ADE type lattice $\mathrm{span}_\mathbb{Z}(E_s^{(j)})$ for fixed $s$. Notice that this calculation only depends on $G$.

We define the elements dual to the point and the volume of $X$ as
\be
\hat v := \sqrt{G} v\,, \qquad \hat v^0 :=\frac{1}{|G|}v^0 - \frac{1}{|G|} B_{|G|} -\frac{\la B_{|G|},B_{|G|}\ra}{2|G|^2} \sqrt{|G|} v\,.
\ee
These elements also define the translation between $x$ and its geometric interpretation as in (\ref{eq:3.18}). We can now finish the construction of $H^{\text{even}}(X,\mathbb{Z})$.
\begin{lem}[Lemma 3.1, \cite{Wen}]
The lattice $\Gamma^{20,4}=H^{\text{even}}(X,\mathbb{Z})$ is generated as
\be\label{eq:3.33}
\Gamma^{20,4}=\mathrm{span}_\mathbb{Z}\left(M_{|G|} \cup \{\,E-\la E,\hat v^0\ra\hat v \mid E \in \Pi_{|G|}\,\}\right)\,.
\ee
\end{lem}
\begin{thm}[Theorem 3.3, \cite{Wen}]
$G=\mathbb{Z}_M, M\in\{2,3,4,6\}$ orbifolding on $T^4$ induces a map on the geometric interpretation as
\be\label{eq:3.34}
(\Sigma_T,V_T,B_T)\longrightarrow(\Sigma,V,B)\,,
\ee
where $\Sigma$ is found as described in (\ref{eq:3.29}), the volume is $V=\frac{V}{|G|}$, and the background is $B=\frac{1}{\sqrt{G}}B_T+\frac{1}{|G|}B_{|G|}-2\hat v$.
\end{thm}

In our Mathematica package, we have included a function that computes the lattice $H^{\text{even}}(X,\mathbb{Z})$ and returns the Gram matrix of the negative definite lattice $x^\perp \cap H^{\text{even}}(X,\mathbb{Z})$ by
\begin{verbatim}
    OrbifoldLattice[L,B,M]
\end{verbatim}
where $L$ is the torus lattice, $B$ is an antisymmetric matrix with elements given in terms of $dx_i\wedge dx_j$,\footnote{The basis change to $d\lambda_i \wedge d\lambda_j$ given in \eqref{eq:3.28} is handled by the Mathematica program.} and $M$ is the order of the cyclic orbifolding symmetry taking values in $\{2,3,4\}$. Note that for $M=3$ and $M=4$, both the choice and the order of basis for the lattice $L$ matters, as the orbifold symmetry generator $g$ should act on the coordinates as
\begin{align}
    g \in \mathbb{Z}_3 & : \, (\lambda_1,\lambda_2,\lambda_3,\lambda_4)\mapsto (\lambda_2-\lambda_1,-\lambda_1,-\lambda_4,\lambda_3-\lambda_4)\,,\\
    g \in \mathbb{Z}_4 & : \, (\lambda_1,\lambda_2,\lambda_3,\lambda_4)\mapsto (\lambda_2,-\lambda_1,-\lambda_4,\lambda_3)\,.
\end{align}
Since only the $\mathbb{Z}_M$-invariant part of cohomology is sent to the K3 cohomology, we need $B\in H^2(T,\mathbb{R})^{\mathbb{Z}_M}$. More specifically,
\begin{align}
\begin{split}
    \mathbb{Z}_3 : \, B \in \mathrm{span}_\mathbb{R} (d\lambda_1\wedge d\lambda_2, 
    d\lambda_3 \wedge d\lambda_4, d\lambda_1\wedge d\lambda_3+d\lambda_4\wedge d\lambda_2, \\
    -d\lambda_1\wedge d\lambda_3 + d\lambda_1\wedge d\lambda_4+d\lambda_2\wedge d\lambda_3 )\,,
\end{split}
\end{align}
\begin{align}
\begin{split}
    \mathbb{Z}_4 : \, B \in \mathrm{span}_\mathbb{R} (d\lambda_1\wedge d\lambda_2, d\lambda_3\wedge d\lambda_4,d\lambda_1 \wedge d\lambda_4 +d\lambda_2\wedge d\lambda_3, \\
    d\lambda_1\wedge d\lambda_3 + d\lambda_4 \wedge d\lambda_2)\,.
\end{split}
\end{align}
This is a condition required for consistent orbifolding. 

\begin{table}[t]
 \begin{center}
 \begin{tabular}{||c c c||} 
 \hline
 HM Lattice & $\CN=(4,4)$ Symmetry Group & Torus Model \\ [0.5ex] 
 \hline\hline
 HM99 & $2^8\cln M_{20} $& $D_4/\mathbb{Z}_2$\\ 
 \hline
 HM100 & $2^{1+8}\cln A_5$ & $A_4/\mathbb{Z}_2$ \\
 \hline
 HM101 & $3^4\cln A_6$ & $D_4/\mathbb{Z}_3$ \\
 \hline
 HM103 & $2^{1+8}\cln S_4$ & $A_1A_3/\mathbb{Z}_2$ \\
 \hline
 HM104 & $2^{1+8}\cln 2.(3\times 3)$ & $A_2^2/\mathbb{Z}_2$ \\
 \hline
 HM105 & $2^{1+8}\cln 2\times S_3$ & $A_1^2A_2/\mathbb{Z}_2$ \\
 \hline
 HM107 & $2^{1+8}\cln 2^3$ &$ A_1^4/\mathbb{Z}_2$ \\
 \hline
  HM107 & $2^{1+8}\cln 2^3$ & $D_4/\mathbb{Z}_4$ \\
 \hline
 HM109 & $3^{1+4}\cln Q_8$ & $A_1^2A_2/\mathbb{Z}_3$ \\
 \hline
 HM114 & $3^{1+4}\cln 4$ & $A_2^2/\mathbb{Z}_3$ \\
 \hline
  HM116 & $(2\times 4^2)\cln S_4$ & $A_1^2A_2/\mathbb{Z}_4$ \\
 \hline
 HM116 & $(2\times 4^2)\cln S_4$ & $A_1^4/\mathbb{Z}_4$ \\[1ex] 
 \hline
\end{tabular}
\end{center}
 \caption{First column lists all the CSS models for which we have found a geometric orbifold dual. For each of them, in the second column we list the symmetry group fixing the $\CN=(4,4)$ algebra, i.e. $\Fix(\mathrm{HM}\#)$. In the last column, we provide the corresponding geometric orbifold model using the torus models from section 4.4 of \cite{Vol}. }\label{Table:1}
\end{table}

In Table \ref{Table:1}, we report the geometric $\mathbb{Z}_M$ orbifolds with $M\in\{2,3,4\}$ of the torus models in \cite{Vol} whose lattices $x^\perp \cap \Gamma^{20,4}$ match with a H{\"o}hn-Mason coinvariant lattice $\mathfrak{F}_L^\perp$, hence finding an explicit duality between twelve type IIA K3 models and heterotic CSS models. We point out the following new results. We recognize two intersections in the moduli space of torus orbifolds from the table: $A_1^4/\mathbb{Z}_2$ and $D_4/\mathbb{Z}_4$, as well as $A_1^2A_2/\mathbb{Z}_4$ and $A_1^4/\mathbb{Z}_4$ meet. Comparing with section 4 of \cite{GHV}, we see that $D_4/\mathbb{Z}_3$ meets the Gepner model $(1)^6$. We have also further specified the structure of $\CN=(4,4)$ symmetry groups of HM100, HM103, HM104, HM105, HM109, HM114 using arguments similar to those in section \ref{sec:Z3 model} and section \ref{sec:Z2 model}.

Lastly, we mention a method for constructing more candidate torus models using the even cohomology formulation. Given an even lattice $\mathfrak{L}$ of rank 4, we can use the gluing construction described in Lemma \ref{lem:gluing lemma} on another copy of the lattice to get the even unimodular lattice
\be\label{eq:3.41}
\Gamma^{4,4}=\coprod_{[r]\in \CD(\mathfrak{L})}(r,r)+\mathfrak{L}(-1) \oplus \mathfrak{L}\,,
\ee
which is necessarily isometric to $\mathrm{II}^{4,4}\cong H^{\text{even}}(T^4,\mathbb{Z})$. Therefore, we can simply define $x$ to be the $\mathbb{R}$-span of $(0,\mathfrak{L})\subset\Gamma^{4,4}$. Let the bilinear form on some basis $g_i$ of $\Gamma^{4,4}$ be given by the matrix $G$ with $G_{ij}:=\la g_i,g_j\ra$. By our arguments, there should be a unimodular matrix $W$ acting as a change of basis such that \be 
W^T G W = \begin{pmatrix}
0 & 1 \\
1 & 0
\end{pmatrix}^{\oplus 4}\,.\ee 
Then $x_T$ in terms of $d\lambda_i$ is expressed as the $\mathbb{R}$-span of $W^{-1}(0,\mathfrak{L})$.

We determine $W$ with the following recipe. Choose a primitive element $w_1$ such that $\la w_1 , w_1 \ra =0$. We know that $w_1\in \mathfrak{L}\cap w_1^\perp$ by construction, so we can project it out by considering $W_1:=(\mathfrak{L}\cap w_1^\perp)/w_1$. More concretely, we choose a set of basis for $\mathfrak{L}\cap w_1^\perp$ containing $w_1$ and then we delete it from the set. We define $w_2$ as the element satisfying $w_2\perp W_1$ and $\la w_1,w_2 \ra=1$. We can now see that $w_1,w_2$ form a standard basis of $U$ and also that $\mathrm{span}_\mathbb{Z}(w_1,w_2)=U \perp W_1$. We continue iteratively until we find $w_1,w_2,\dots,w_7,w_8$ so that they form a standard basis of $U^{\oplus 4}$. Then $w_i$ constitute the columns of $W$.

%--------------------------------------------------------------------------------------------------

\section{Conclusions and open problems}\label{sec:Conclude}

We have a presented a number of explicit constructions of dual pairs of string theories which have supersymmetry preserving symmetry groups that are subgroups of the Conway group. In each example we were able to not only match symmetry groups but we also matched points in
the moduli space of $T^4$ heterotic and K3 type II compactifications. We also provide Mathematica code that automates the required computations. 

There are two interesting open problems that we feel it would be interesting to address. On the heterotic side the full
symmetry group of the $T^4$ compactification involves an extension of the lattice automorphism group by $\IZ_2^{24}$. This
occurs because associativity of the vertex operator OPEs for vertex operators that create states with lattice momenta requires the addition
of cocycle factors. Essentially one must work with a projective representation of the Narain lattice. This extension is only visible when we study products of vertex operators involving lattice momenta, that is when we look at the addition law for lattice momenta. On the type II side
this corresponds to addition of D-brane charge and so is not visible in perturbation theory. Duality requires that there exist a projective representation of the cohomology lattice on the type II side that governs addition of D-brane charge. While it is known that in general D-brane charges should be described by K-theory rather than ordinary cohomology \cite{Witten:1998cd}, it is not clear to us that the origin of these $\IZ_2$ factors has been understood in the literature. In particular it does not follow from twisted K-theory since $H^3(X, \IZ)$ is trivial for K3 surfaces $X$. 

A second interesting question which this work might help to address is the claim that all exact symmetries in string theory are gauge
symmetries. There are a variety of arguments behind this claim, but, as far as we know, none of them really address the question of whether discrete symmetries of string theory in flat space are always gauge symmetries. This work provides dual pairs of string models with
large discrete symmetry groups and they seems like a good starting point to investigate how one would prove or disprove that these discrete symmetries are gauge symmetries. 

%-----------------------------------------------------------------------------------------------
\appendix
\section{Details of CSS lattice constructions}\label{app:A}

As a supplement to section \ref{sec:Construct Theory}, in our Mathematica package we provide
\begin{itemize}
    \item the Gram matrix $G$ of the lattice HM\#,
    \item the lattices $\mathfrak{F}_L,\mathfrak{F}_R$ embedded in $E_8,\Lambda$, their orthogonal complements $\mathfrak{F}_L^\perp,\mathfrak{F}_R^\perp$, and their duals $\left(\mathfrak{F}_L^\perp\right)^\vee,\left(\mathfrak{F}_R^\perp\right)^\vee$,
    \item the discriminant group $\CD(\mathfrak{F})$ of the lattices, which is isomorphic for all lattices considered,
    \item isometries between the discriminant groups $\CD(\mathfrak{F}_L^\perp),\CD(\mathfrak{F}_L),\CD(\mathfrak{F}_R),\CD(\mathfrak{F}_R^\perp)$ in terms of the generator glue vectors $c_i,d_i,e_i,f_i$ of the discriminant groups respectively,
    \item the generator glue vectors $(r,\psi(r))$ as in (\ref{eq:2.31}).
\end{itemize}
The glue vectors are generated from $(c_i,\psi(c_i))$ with $[c_i]$ a generator of $\mathbb{Z}_{d_i}\subset \CD(\mathfrak{F}_R^\perp)= \prod_j \mathbb{Z}_{d_j}$, and $\psi$ is found by chasing the isometries provided.

As a supplement to section \ref{sec:K3 lattice}, we provide a function which returns the Gram matrix of the lattice $x^\perp \cap H^{\mathrm{even}}(x,\mathbb{Z})$ given a torus model. For a complete list of commands with explanations, the user can run
\begin{verbatim}
    ?CSSCompactifications`*
\end{verbatim}
after the package is installed as instructed in the Mathematica notebook.

In the rest of the appendix, we will describe our procedure in detail for computing the CSS lattice $\Gamma$. We give an overview of the steps:
\begin{steps}
    \item Find an isometric embedding $\mathfrak{F}_R\subset E_8$ of $\mathfrak{F}_L$ by computing the Gram matrices of sublattices of the same rank in $E_8$.
    \item Check if the embedding $\mathfrak{F}_R\subset E_8$ is primitive.
    \item Put all the lattices $\mathfrak{F}_L,\mathfrak{F}_L^\perp,\mathfrak{F}_R,\mathfrak{F}_R^\perp$ in Smith Normal Form. Get the discriminant group by (\ref{eq:2.16}) and the generators of each lattice's discriminant groups by (\ref{eq:2.14}).
    \item Find isometries $\CD(\mathfrak{F}_L^\perp)\rightarrow \CD(\mathfrak{F}_L)$ and $\CD(\mathfrak{F}_R)\rightarrow\CD(\mathfrak{F}_R^\perp) $. (Recall that the isometry $\CD(\mathfrak{F}_L)\rightarrow \CD(\mathfrak{F}_R)$ is induced by $\mathfrak{F}_L\cong \mathfrak{F}_R$.)
    \item Construct the generators of $\Gamma$ by considering the Smith Normal Form basis $\tilde V_L$ of $\mathfrak{F}_L^\perp$. Specifically, take the elements $(v^L_i,0)$ for $d_i=1$ and $(v^L_i,\psi(v^L_i))$ for $d_i\neq 1$ where $v_i^L$ denotes the $i$th basis vector in $\tilde V_L$.
\end{steps}
\subsection{Finding an embedding in \texorpdfstring{$E_8$}{E8}}

To find an embedding of the sublattice $\mathfrak{F}_L$ in $E_8$ with the Gram matrix $G$, we use two methods according to how large $\mathfrak{F}_L$ is. 

The first method is applicable if all the basis vectors in $\mathfrak{F}_L$ satisfy $\|v\|^2\leq 8$, i.e. the diagonal entries of $G$ are $G_{ii}\leq 8$. The advantage of this method is that it is faster than brute force computation, but it is not applicable for all sublattices $\mathfrak{F}_L$.

We calculate all the inner products of vectors in $E_8$ with norm $\|v\|^2\leq 8$ with each other, and store them in a database. We define a function 
\begin{verbatim}
    VectorData[v,m,x]
\end{verbatim}
which takes in a vector \verb|v|, and returns all the vectors from the database which have inner product with \verb|v| equal to \verb|x| with \verb|x|$=0,\pm1,\pm2,\pm3,\pm4$ and has norm squared \verb|m|.

Now, to find a set of vectors in $E_8$ with the same Gram matrix $G$, we implement the following algorithm:
\begin{steps}
\item Choose a random vector $v_1$ with $\|v_1\|^2=G_{11}$ and add it to the list.
\item Chosen vector $v_{n-1}$, choose a random vector $v_n$ from the set 
$$\verb|VectorData|(v_1,G_{nn},G_{1n}) \cap \dots \cap \verb|VectorData|(v_{n-1},G_{nn},G_{(n-1)n})$$ 
and add it to the list. If the set is empty, choose a different $v_{n-1}$. If all such sets are empty, choose a different $v_{n-2}$ and so on.
\item Repeat until the list is complete.
\end{steps}
This way, we can find the embedding quickly if it exists.

We use the second method only if there is a basis vector with $\|v\|^2>8$ in the Gram matrix of $\mathfrak{F}_L$. It is a brute force computation to find a vector set with the desired inner products.
\begin{steps}
    \item Generate a random vector in the lattice such that $v_1=\sum c_i f_i, c_i\leq k$, where $f_i$ are a basis for $E_8$ and $k$ is a fixed upper bound for the coefficients according to $\|v_1\|^2=G_{11}$.
    \item If vector $v_1$ has norm $G_{11}$, add it to the list, and generate random vector $v_2$.
    \item Chosen $v_1,\dots,v_{n-1}$, generate a random vector $v_n$. If $v_n$ satisfies $v_i \cdot v_n = G_{in}$ for $ i \leq n$, add this vector to the list. If there are no such $v_n$, choose a different $v_{n-1}$ and so on.
    \item Repeat until the list is complete.
\end{steps}

For $n\geq 5$ this problem becomes computationally complex. Therefore, for sublattices with higher ranks, we compute the set of all possible triples $(v_1,v_2,v_3)$ that have inner products given by $G_{ij}$ with $i,j \leq 3$, and try random vectors for all the others until we get the desired embedding.

\subsection{Primitivity}

Once we find an embedding in $E_8$, we would like to find out whether this sublattice is primitive or not. Recall that a sublattice $\mathfrak{G}$ is primitive in $\mathfrak{L}$ if $\mathfrak{L}/\mathfrak{G}$ is free. The following is a more computer friendly definition we use.

\begin{prop}
Let $\mathfrak{G}$ be a sublattice in $\mathfrak{L}$. Let $\{v_1, \dots, v_k\}$ be the basis of $\mathfrak{G}$ and $\{f_1,\dots,f_n\}$ be the basis of $\mathfrak{L}$. Write $v_i = \sum c_{ij} f_j$, or in matrix form, $V=CF$ where $v_i$ and $f_i$ are rows of $V$ and $F$ respectively. 

Then $\mathfrak{G}$ is primitive if the GCD of all $k \times k$ minors of C is equal to 1.
\end{prop}

\begin{proof}
We use a fact from \cite{DQ}: the $k \times n$ matrix $C$ can be completed to a $n \times n$ unimodular matrix if and only if the GCD of all $k \times k$ minors of C is equal to 1. 

Suppose $\mathfrak{L}/\mathfrak{G}$ is not free. Let $D$ be the $(n-k) \times n$ matrix that completes $C$ to a unimodular matrix $M$. Since $M F$ is a basis for of $\mathfrak{L}$ and $C F$ is a basis for $\mathfrak{G}$, we have $D F$ as a basis for $\mathfrak{L}/\mathfrak{G}$. If $\mathfrak{L}/\mathfrak{G}$ is not free, then some linear combination of the row vectors in $D F$ is in $\mathfrak{G}$. But that means vectors in $D$ and $C$ are not linearly independent, therefore together they cannot form a unimodular matrix, leading to a contradiction.
\end{proof}

\subsection{Isometries between the discriminant groups}

We would like to find an isometry 
\be
(\mathcal{D}(\mathfrak{L}),\bar q)\rightarrow (\mathcal{D}(\mathfrak{L}^{\perp}),-\bar q)\,.
\ee
Since $\CD(\mathfrak{L})\cong\CD(\mathfrak{L}^\perp)$, we choose an endomorphism $\psi$, We first check if it is an isomorphism, and then check if $\bar q=-\bar q \psi$.

We are looking for the automorphisms of $\mathbb{Z}_{d_1} \times \dots \times \mathbb{Z}_{d_k}$. To characterize such an automorphism, it is enough to specify where the generators map to. We collect all the elements with order $d_i$ in set $D_i$, and map each $1_{d_i}$ to some element in $D_i$. The mapping will take the form
\be 1_{d_i} \mapsto a_{i1} 1_{d_1} + \dots + a_{ik} 1_{d_k}\,.\ee

A priori, such a map is an endomorphism. There are two methods to check if it is an automorphism. First, we can check if each $1_{d_i}$ is in the image of the mapping. We characterize an endomorphism by putting the coefficients in a matrix
\be A= 
\begin{pmatrix}
a_{11} & \dots & a_{1k}\\
\vdots \\
a_{k1} & \dots & a_{kk}
\end{pmatrix}\,.\ee
We consider the inverse
\be A^{-1} = \begin{pmatrix}
\frac{p_{11}}{q_{11}} & \dots & \frac{p_{1k}}{q_{1k}}\\
\vdots \\
\frac{p_{k1}}{q_{k1}} & \dots & \frac{p_{kk}}{q_{kk}}
\end{pmatrix}\,,\ee
with the matrix entries rational. We take the $i$th row, multiply it by $q=\text{lcm}(q_{i1},\dots,q_{ik})$ to clear out the denominator, and get
\be q  1_{d_i} = n_{i1} 1_{d_1} + \dots + n_{ik} 1_{d_k}\,.\ee
If $q$ is coprime with $d_i$, then $1_{d_i}$ is in the image of the endomorphism. If this is true for all $i$, then this mapping is an automorphism.

The second method is a simpler computation. For each endomorphism
\be \psi:\mathbb{Z}_{d_1} \times \dots \times \mathbb{Z}_{d_k}\rightarrow \mathbb{Z}_{d_1} \times \dots \times \mathbb{Z}_{d_k}\,,\ee
we calculate the number of elements in the image, and then we calculate the order of each element in the image. This characterizes the image as a unique finite abelian group, which we compare with the preimage to see if they are the same group. 

Finally, we check if the isomorphism $\psi$ is an isometry by the straightforward computation
\be \bar q \overset{?}{=}-\bar q \psi\,,\ee
and repeat until we obtain an isometry.

\section{Lattice construction data for HM44}\label{app:B}

In this appendix, we show explicitly how CSS lattices are constructed, using HM44 as our guiding example. Running \begin{verb} CSSLatticeConstruction[44] \end{verb} in our Mathematica package reproduces the content we present here. 

We use the rows of the following matrix as basis vectors for the $E_8$ lattice in $\mathbb{R}^8$:
\begin{align}
[E_8]=\left(
\begin{array}{cccccccc}
2 & 0 & 0 & 0 & 0 & 0 & 0 & 0 \\
-1 & 1 & 0 & 0 & 0 & 0 & 0 & 0 \\
0 & -1 & 1 & 0 & 0 & 0 & 0 & 0 \\
0 & 0 & -1 & 1 & 0 & 0 & 0 & 0 \\
0 & 0 & 0 & -1 & 1 & 0 & 0 & 0 \\
0 & 0 & 0 & 0 & -1 & 1 & 0 & 0 \\
0 & 0 & 0 & 0 & 0 & -1 & 1 & 0 \\
\frac{1}{2} & \frac{1}{2} & \frac{1}{2} & \frac{1}{2} & \frac{1}{2} & \frac{1}{2} & \frac{1}{2} & \frac{1}{2} \\
\end{array}
\right)\,.
\end{align}
In our notation, we put brackets around the lattice to denote its basis vectors. For the Leech lattice $\Lambda$, the basis vectors in $\mathbb{R}^{24}$ are given as the rows of the matrix
\begin{align}
[\Lambda]=c\left(
\begin{array}{cccccccccccccccccccccccc}
8 & 0 & 0 & 0 & 0 & 0 & 0 & 0 & 0 & 0 & 0 & 0 & 0 & 0 & 0 & 0 & 0 & 0 & 0 & 0 & 0 & 0 & 0 & 0 \\
4 & 4 & 0 & 0 & 0 & 0 & 0 & 0 & 0 & 0 & 0 & 0 & 0 & 0 & 0 & 0 & 0 & 0 & 0 & 0 & 0 & 0 & 0 & 0 \\
4 & 0 & 4 & 0 & 0 & 0 & 0 & 0 & 0 & 0 & 0 & 0 & 0 & 0 & 0 & 0 & 0 & 0 & 0 & 0 & 0 & 0 & 0 & 0 \\
4 & 0 & 0 & 4 & 0 & 0 & 0 & 0 & 0 & 0 & 0 & 0 & 0 & 0 & 0 & 0 & 0 & 0 & 0 & 0 & 0 & 0 & 0 & 0 \\
4 & 0 & 0 & 0 & 4 & 0 & 0 & 0 & 0 & 0 & 0 & 0 & 0 & 0 & 0 & 0 & 0 & 0 & 0 & 0 & 0 & 0 & 0 & 0 \\
4 & 0 & 0 & 0 & 0 & 4 & 0 & 0 & 0 & 0 & 0 & 0 & 0 & 0 & 0 & 0 & 0 & 0 & 0 & 0 & 0 & 0 & 0 & 0 \\
4 & 0 & 0 & 0 & 0 & 0 & 4 & 0 & 0 & 0 & 0 & 0 & 0 & 0 & 0 & 0 & 0 & 0 & 0 & 0 & 0 & 0 & 0 & 0 \\
2 & 2 & 2 & 2 & 2 & 2 & 2 & 2 & 0 & 0 & 0 & 0 & 0 & 0 & 0 & 0 & 0 & 0 & 0 & 0 & 0 & 0 & 0 & 0 \\
4 & 0 & 0 & 0 & 0 & 0 & 0 & 0 & 4 & 0 & 0 & 0 & 0 & 0 & 0 & 0 & 0 & 0 & 0 & 0 & 0 & 0 & 0 & 0 \\
4 & 0 & 0 & 0 & 0 & 0 & 0 & 0 & 0 & 4 & 0 & 0 & 0 & 0 & 0 & 0 & 0 & 0 & 0 & 0 & 0 & 0 & 0 & 0 \\
4 & 0 & 0 & 0 & 0 & 0 & 0 & 0 & 0 & 0 & 4 & 0 & 0 & 0 & 0 & 0 & 0 & 0 & 0 & 0 & 0 & 0 & 0 & 0 \\
2 & 2 & 2 & 2 & 0 & 0 & 0 & 0 & 2 & 2 & 2 & 2 & 0 & 0 & 0 & 0 & 0 & 0 & 0 & 0 & 0 & 0 & 0 & 0 \\
4 & 0 & 0 & 0 & 0 & 0 & 0 & 0 & 0 & 0 & 0 & 0 & 4 & 0 & 0 & 0 & 0 & 0 & 0 & 0 & 0 & 0 & 0 & 0 \\
2 & 2 & 0 & 0 & 2 & 2 & 0 & 0 & 2 & 2 & 0 & 0 & 2 & 2 & 0 & 0 & 0 & 0 & 0 & 0 & 0 & 0 & 0 & 0 \\
2 & 0 & 2 & 0 & 2 & 0 & 2 & 0 & 2 & 0 & 2 & 0 & 2 & 0 & 2 & 0 & 0 & 0 & 0 & 0 & 0 & 0 & 0 & 0 \\
2 & 0 & 0 & 2 & 2 & 0 & 0 & 2 & 2 & 0 & 0 & 2 & 2 & 0 & 0 & 2 & 0 & 0 & 0 & 0 & 0 & 0 & 0 & 0 \\
4 & 0 & 0 & 0 & 0 & 0 & 0 & 0 & 0 & 0 & 0 & 0 & 0 & 0 & 0 & 0 & 4 & 0 & 0 & 0 & 0 & 0 & 0 & 0 \\
2 & 0 & 2 & 0 & 2 & 0 & 0 & 2 & 2 & 2 & 0 & 0 & 0 & 0 & 0 & 0 & 2 & 2 & 0 & 0 & 0 & 0 & 0 & 0 \\
2 & 0 & 0 & 2 & 2 & 2 & 0 & 0 & 2 & 0 & 2 & 0 & 0 & 0 & 0 & 0 & 2 & 0 & 2 & 0 & 0 & 0 & 0 & 0 \\
2 & 2 & 0 & 0 & 2 & 0 & 2 & 0 & 2 & 0 & 0 & 2 & 0 & 0 & 0 & 0 & 2 & 0 & 0 & 2 & 0 & 0 & 0 & 0 \\
0 & 2 & 2 & 2 & 2 & 0 & 0 & 0 & 2 & 0 & 0 & 0 & 2 & 0 & 0 & 0 & 2 & 0 & 0 & 0 & 2 & 0 & 0 & 0 \\
0 & 0 & 0 & 0 & 0 & 0 & 0 & 0 & 2 & 2 & 0 & 0 & 2 & 2 & 0 & 0 & 2 & 2 & 0 & 0 & 2 & 2 & 0 & 0 \\
0 & 0 & 0 & 0 & 0 & 0 & 0 & 0 & 2 & 0 & 2 & 0 & 2 & 0 & 2 & 0 & 2 & 0 & 2 & 0 & 2 & 0 & 2 & 0 \\
-3 & 1 & 1 & 1 & 1 & 1 & 1 & 1 & 1 & 1 & 1 & 1 & 1 & 1 & 1 & 1 & 1 & 1 & 1 & 1 & 1 & 1 & 1 & 1 \\
\end{array}
\right)\,,
\end{align}
where $c=1/\sqrt{8}$. 

Among rank 6 Höhn-Mason sublattices, only HM44 has an embedding $\mathfrak{F}_R$ in $E_8$. For the construction of the CSS lattice, we use its orthogonal complement, which is provided by the Mathematica package as
\begin{align}
[\mathfrak{F}_R^\perp]=\left(
\begin{array}{cccccccc}
0 & 1 & 1 & -2 & 0 & 1 & -3 & -2 \\
0 & 0 & 2 & -2 & 3 & -1 & -1 & -1 \\
\end{array}
\right)\,.
\end{align}
Similarly, we take the orthogonal complement of the corresponding Leech sublattice, which has basis vectors given as
\begin{align}
[\mathfrak{F}_L^\perp]=c\left(
\begin{smallmatrix}
0 & 0 & 0 & 2 & -2 & 0 & 0 & 0 & 2 & 0 & 0 & 0 & 0 & 0 & 0 & -2 & 0 & 0 & 2 & 0 & -2 & 0 & -2 & 2 \\
-1 & -1 & 1 & -1 & 1 & 1 & -1 & 1 & -1 & -1 & -1 & 1 & 1 & 1 & 1 & -1 & -1 & 1 & -1 & -1 & 1 & -1 & 1 & -3 \\
0 & 0 & 0 & 0 & 2 & 0 & 2 & 0 & 0 & 0 & 0 & 0 & -2 & 0 & -2 & 0 & 2 & 0 & 0 & 2 & 2 & 2 & 0 & 0 \\
-2 & 0 & 0 & -2 & 0 & 0 & 0 & 0 & 0 & 2 & -4 & 2 & 0 & 0 & 2 & 2 & 0 & 0 & 0 & 0 & 2 & 0 & 0 & -2 \\
-1 & -1 & -1 & -1 & -1 & -1 & -1 & -1 & 1 & 1 & -1 & -1 & 3 & -1 & 1 & 1 & -1 & -1 & 1 & 1 & 1 & 1 & -1 & -1 \\
1 & -1 & 1 & 1 & 1 & -1 & 1 & 1 & -1 & -1 & 1 & -1 & 1 & 1 & -1 & -3 & 1 & 1 & 1 & -1 & 1 & 1 & -3 & 3 \\
-1 & -1 & 1 & 1 & 1 & -1 & 1 & -1 & -1 & 1 & -1 & 1 & 1 & 1 & -1 & -1 & -1 & 1 & -1 & 1 & -1 & -1 & 1 & -3 \\
-1 & -1 & 1 & 1 & 1 & -1 & 1 & -1 & 1 & 1 & -1 & -1 & 1 & -1 & 1 & -1 & 1 & 1 & 3 & -1 & -1 & 1 & -1 & 1 \\
0 & 0 & 0 & 0 & 0 & -2 & 2 & 0 & 0 & 0 & -2 & 2 & 0 & 2 & 0 & -2 & 0 & 0 & 0 & 0 & 0 & 2 & -2 & 0 \\
0 & 0 & 0 & 0 & 0 & 0 & -2 & 2 & 0 & -2 & 2 & 0 & 2 & 0 & 2 & 0 & -2 & -2 & 0 & 0 & 0 & 0 & 0 & 0 \\
1 & 1 & -1 & -1 & 1 & -1 & 1 & -1 & 1 & 1 & 1 & 1 & 1 & -1 & -1 & 1 & -1 & -1 & -1 & 3 & -1 & 1 & 1 & -1 \\
0 & 0 & 0 & 2 & -2 & 0 & 0 & 0 & 0 & 0 & 2 & 0 & 0 & 2 & 0 & 0 & -2 & 0 & 0 & 0 & -2 & -2 & 2 & 0 \\
-1 & 1 & 1 & -1 & 1 & 1 & 1 & 1 & -1 & 1 & -1 & 1 & -1 & -1 & 1 & 1 & 1 & 1 & 1 & 1 & 3 & 1 & 1 & -1 \\
0 & 0 & 0 & -2 & 2 & 0 & 0 & 0 & 0 & 0 & -2 & 0 & 2 & 0 & 2 & 2 & -2 & 0 & 0 & 0 & 0 & 0 & 0 & -2 \\
1 & -1 & 1 & 1 & -1 & -1 & 1 & -1 & 1 & -1 & -1 & -1 & 1 & 1 & 1 & -1 & -1 & 1 & 1 & -3 & -1 & -1 & -1 & 1 \\
-1 & -1 & 1 & -1 & 1 & -1 & 1 & 1 & -1 & 1 & -1 & -1 & 1 & 1 & -1 & 1 & 1 & 1 & -1 & 1 & 3 & 1 & -1 & -1 \\
0 & 0 & 0 & 0 & 0 & 0 & 0 & 0 & 0 & 0 & 0 & 0 & 0 & 0 & 0 & 0 & 2 & -2 & 2 & 2 & 2 & 2 & -2 & 2 \\
1 & -1 & 1 & 1 & 1 & -1 & 1 & 1 & -1 & -1 & 1 & -1 & 1 & 1 & -1 & -3 & 1 & 1 & 1 & -1 & 1 & 1 & 1 & -1 \\
\end{smallmatrix}
\right)\,.
\end{align}

In order to glue the two together as in \eqref{eq:2.31}, we need the correspondence between their discriminant groups. Putting the lattices in Smith Normal Form as in \eqref{eq:2.15} immediately shows that the discriminant groups are isomorphic to $\mathbb{Z}_{10}\times\mathbb{Z}_{30}$. The generators that correspond to $1\in\mathbb{Z}_{10}$ and $1\in\mathbb{Z}_{30}$ are given for each sublattice as follows:
\begin{align}
\begin{split}
&[\mathcal{D}(\mathfrak{F}_L^\perp)]=\begin{pmatrix}
c_1 \\
c_2
\end{pmatrix}\\
&= c\left(
\begin{smallmatrix}
\frac{41}{10} & \frac{1}{10} & \frac{9}{10} & -\frac{3}{10} & \frac{1}{10} & \frac{13}{10} & -\frac{13}{10} & -\frac{21}{10} & -\frac{19}{10} & -\frac{1}{10} & -\frac{3}{2} & \frac{11}{10} & \frac{3}{2} & \frac{21}{10} & \frac{1}{10} & -\frac{9}{10} & \frac{31}{10} & -\frac{37}{10} & \frac{7}{2} & \frac{11}{10} & \frac{1}{10} & -\frac{3}{2} & \frac{3}{2} & -\frac{1}{2} \\
-\frac{7}{5} & -\frac{3}{5} & -\frac{33}{5} & \frac{34}{15} & 2 & -\frac{4}{3} & 1 & 2 & 0 & 1 & -\frac{8}{15} & \frac{2}{15} & -\frac{7}{3} & -4 & -\frac{14}{15} & -\frac{28}{15} & -\frac{27}{5} & \frac{8}{5} & -\frac{2}{15} & -\frac{38}{15} & 2 & -\frac{1}{3} & \frac{8}{15} & \frac{8}{15} \\
\end{smallmatrix}
\right)\,,
\end{split}\\
&[\mathcal{D}(\mathfrak{F}_R^\perp)]=\begin{pmatrix}
f_1 \\
f_2
\end{pmatrix}
=\left(
\begin{array}{cccccccc}
0 & 0 & \frac{1}{5} & -\frac{1}{5} & \frac{3}{10} & -\frac{1}{10} & -\frac{1}{10} & -\frac{1}{10} \\
0 & -\frac{1}{30} & \frac{1}{10} & -\frac{1}{15} & \frac{1}{5} & -\frac{1}{10} & \frac{1}{30} & 0 \\
\end{array}
\right)\,.
\end{align}
We found the following isometry between the two
\begin{align}
\bar{\psi}:\mathcal{D}(\mathfrak{F}_L^\perp)&\to \mathcal{D}(\mathfrak{F}_R^\perp)\\
(c_1,c_2) &\mapsto (9 f_1, 4f_1 + 19f_2)\,.
\end{align}
Using this isometry, we can obtain the glue vectors $(r,\psi(r))$ for the CSS lattice and finish the construction to get
\begin{align}
\Gamma = \coprod_{\substack{n \in\mathbb{Z}_{10},\\m\in\mathbb{Z}_{30}}} n(c_1, 9f_1) + m(c_2, 4f_1+19f_2) + (\mathfrak{F}_L^\perp,\mathfrak{F}_R^\perp)\,.
\end{align}


\begin{thebibliography}{99}

  
 %----------INTRO-------------

\bibitem{Harvey:2017xdt}
J.A.~Harvey and G.W.~Moore, \textit{Conway subgroup symmetric compactifications of heterotic string},
\textit{J.\ Phys.\ A} \textbf{51} (2018) 354001
[\hyperlink{https://arxiv.org/abs/1712.07986}{\texttt{arXiv:1712.07986}}].
%doi:10.1088/1751-8121/aac9d1

\bibitem{Banerjee:2020szx}
A.~Banerjee and G.W.~Moore, \textit{Hyperk\"ahler isometries of K3 surfaces},
\textit{JHEP} \textbf{12} (2020) 193
[\texttt{\hyperlink{https://arxiv.org/abs/2009.11769}{arXiv:2009.11769}}].
%doi:10.1007/JHEP12(2020)193

%\cite{Banks:2010zn}
\bibitem{Banks:2010zn}
T.~Banks and N.~Seiberg,
\textit{Symmetries and strings in field theory and gravity},
\textit{Phys. Rev. D} \textbf{83} (2011) 084019
[\texttt{\hyperlink{https://arxiv.org/abs/1011.5120}{arXiv:1011.5120}}].
%doi:10.1103/PhysRevD.83.084019

\bibitem{Harlow:2018tng}
D.~Harlow and H.~Ooguri,
\textit{Symmetries in quantum field theory and quantum gravity},
\textit{Commun. Math. Phys.} \textbf{383} (2021) 1669
[\texttt{\hyperlink{https://arxiv.org/abs/1810.05338}{arXiv:1810.05338}}].


%-----------Lattice Theory-----------

\bibitem{Narain:1985jj}
K.S.~Narain,
\textit{New heterotic string theories in uncompactified dimensions $<10$},
\textit{Phys. Lett. B} \textbf{169} (1986) 41.

\bibitem{Narain:1986am}
K.S.~Narain, M.H.~Sarmadi and E.~Witten,
\textit{A note on toroidal compactification of heterotic string theory},
\textit{Nucl. Phys. B} \textbf{279} (1987) 369.


\bibitem{Nikulin} 
V.V.~Nikulin,
\textit{Integral symmetric bilinear forms and some of their applications},
\textit{Izv.\ Akad.\ Nauk SSSR Ser.\ Mat.} \textbf{43} (1979) 111.

\bibitem{Gerstein}
L.J.~Gerstein,
\textit{Basic quadratic forms},
American Mathematical Society, Providence (2008).

\bibitem{Newman}
M.~Newman, 
\textit{Integral matrices}, 
Academic Press, New York (1972).

\bibitem{DQ}
C.C.~McDuffie,
\textit{The theory of matrices},
Chelsea Publ. Co., New York (1946).

\bibitem{HM} 
G. H{\"o}hn and G. Mason, 
\textit{The 290 fixed-point sublattices of the Leech lattice},
\textit{J.\ Algebra} \textbf{448} (2016) 618
[\texttt{\hyperlink{https://arxiv.org/abs/1505.06420}{arXiv:1505.06420}}].

\bibitem{SPLAG}
J.H.~Conway and N.J.A.~Sloane,
\textit{Sphere packings, lattices and groups},
Springer-Verlag, New York (1993).

%-----------F-Theory------------------
\bibitem{ATLAS}
J.H.~Conway et al., 
\textit{ATLAS of finite groups},
Oxford University Press, Oxford (1985).

\bibitem{Xiao}
G.~Xiao, 
\textit{Galois covers between K3 surfaces}, 
\textit{Ann. Inst. Fourier} \textbf{46} (1996) 73.

\bibitem{Mukai} S.~Mukai, 
\textit{Finite groups of automorphisms of K3 surfaces and the Mathieu group}, 
\textit{Inventiones mathematicae} \textbf{94} (1988) 183.

\bibitem{Sarti}
C. Bonnafé and A. Sarti, 
\textit{K3 surfaces with maximal finite automorphism groups containing $M_{20}$},
\texttt{\hyperlink{https://arxiv.org/abs/1910.05955}{arXiv:1910.05955}}.

\bibitem{zhang2006} D.-Q.~Zhang, 
\textit{The alternating groups and K3 surfaces},
\textit{Jounral of Pure and Applied Algebra} \textbf{207} (2006) 119
[\texttt{\hyperlink{https://arxiv.org/abs/math/0506610}{math/0506610}}].

\bibitem{zhang2005}
D.-Q.~Zhang, 
\textit{Automorphisms of K3 surfaces},
\texttt{\hyperlink{https://arxiv.org/abs/math/0506612}{math/0506612}}.

\bibitem{huybrechts2016lectures}
D. Huybrechts, 
\textit{Lectures on K3 surfaces},
Cambridge University Press (2016).

\bibitem{Bershadsky} 
M. Bershadsky, K. Intriligator, S. Kachru, D.R. Morrison, V. Sadov and C. Vafa,
\textit{Geometric singularities and enhanced gauge symmetries},
\textit{Nuclear Physics B} \textbf{481} (1996) 215
[\texttt{\hyperlink{https://arxiv.org/abs/hep-th/9605200}{hep-th/9605200}}].

\bibitem{Shioda} T. Shioda,
\textit{ The Mordell-Weil lattice of $y^2=x^3+t^5-1/t^5-11$},
\textit{ Comment. Math. Univ. St. Pauli} \textbf{56} (2007) 45.


%----------Strings on K3----------------

\bibitem{Hiker}
W.~Nahm and K.~Wendland, 
\textit{A hiker's guide to K3. Aspects of $N = (4, 4)$ superconformal field theory with central charge $c = 6$}, 
\textit{Commun. Math. Phys.} \textbf{216} (2001) 85 
[\texttt{\hyperlink{https://arxiv.org/abs/hep-th/9912067}{hep-th/9912067}}].

\bibitem{gtvw}
M.R.~Gaberdiel, A.~Taormina, R.~Volpato and K.~Wendland,
\textit{A K3 sigma model with $\mathbb{Z}^8_2 \cln \mathbb{M}_{20}$ symmetry},
\textit{JHEP} \textbf{02} (2014) 022
[\texttt{\hyperlink{https://arxiv.org/abs/1309.4127}{arXiv:1309.4127}}].

\bibitem{GHV}
M.R.~Gaberdiel, S.~Hohenegger and R.~Volpato,
\textit{Symmetries of K3 sigma models},
\textit{Commun.\ Num.\ Theor.\ Phys.} \textbf{6} (2012) 1 
[\texttt{\hyperlink{https://arxiv.org/abs/1106.4315}{arXiv:1106.4315}}].

\bibitem{VG}
M.R.~Gaberdiel and R.~Volpato, 
\textit{Mathieu Moonshine and orbifold K3s},
\texttt{\hyperlink{https://arxiv.org/abs/1206.5143}{arXiv:1206.5143}}.

\bibitem{Vol}
R.~Volpato,
\textit{On symmetries of $\mathcal{N}=(4,4)$ sigma models on $T^4$},
\textit{JHEP} \textbf{08} (2014) 094
[\texttt{\hyperlink{https://arxiv.org/abs/1403.2410}{arXiv:1403.2410}}].

\bibitem{Vol2019}
R.~Volpato,
\textit{Some comments on symmetric orbifolds of K3},
\textit{JHEP} \textbf{10} (2019) 082
[\texttt{\hyperlink{https://arxiv.org/abs/1902.11093}{arXiv:1902.11093}}].

\bibitem{Aspinwall:1994rg}
P.S.~Aspinwall and D.R.~Morrison,
\textit{String theory on K3 surfaces},
\textit{AMS/IP Stud.\ Adv.\ Math.} \textbf{1} (1996) 703
[\texttt{\hyperlink{https://arxiv.org/abs/hep-th/9404151}{hep-th/9404151}}].

\bibitem{Aspinwall:1996mn}
P.S.~Aspinwall,
\textit{K3 surfaces and string duality}, 
\texttt{\hyperlink{https://arxiv.org/abs/hep-th/9611137}{hep-th/9611137}}.

%\cite{Eguchi:2010ej}
\bibitem{Eguchi:2010ej}
T.~Eguchi, H.~Ooguri and Y.~Tachikawa, 
\textit{Notes on the K3 surface and the Mathieu group $M_{24}$},
\textit{Exper.\ Math.} \textbf{20} (2011) 91 
[\texttt{\hyperlink{https://arxiv.org/abs/1004.0956}{arXiv:1004.0956}}].
%%doi:10.1080/10586458.2011.544585

%--------Charge Lattice------------------

\bibitem{Wen}
K.~Wendland, 
\textit{Consistency of orbifold conformal field theories on K3}, 
\textit{Adv.\ Theor.\ Math.\ Phys.} \textbf{5} (2002) 429
[\texttt{\hyperlink{https://arxiv.org/abs/hep-th/0010281}{hep-th/0010281}}].

\bibitem{Wen2001}
K.~Wendland, 
\textit{Orbifold constructions of K3: a link between conformal field theory and geometry}, 
\textit{Contemp.\ Math.} \textbf{310} (2002) 333
[\texttt{\hyperlink{https://arxiv.org/abs/hep-th/0112006}{hep-th/0112006}}].

\bibitem{Ino}
H.~Inose, 
\textit{On certain Kummer surfaces which can be realized as
	non-singular quartic surfaces in $\mathbb{P}^3$}, 
\textit{J.\ Fac.\ Sci.\ Univ.\ Tokyo} \textbf{23} (1976) 545.

%------ Conclusions --------

\bibitem{Witten:1998cd}
E.~Witten,
\textit{D-branes and K theory}, 
\textit{JHEP} \textbf{12} (1998) 019
[\texttt{\hyperlink{https://arxiv.org/abs/hep-th/9810188}{hep-th/9810188}}].
%doi:10.1088/1126-6708/1998/12/019

\end{thebibliography}
\end{document}